\documentclass[a4paper]{amsart}
\usepackage{a4wide,amsmath,stmaryrd}
\usepackage[pdftex]{graphicx}

\newtheorem{theorem}{Theorem}
\newtheorem{lemma}[theorem]{Lemma}
\newtheorem{proposition}[theorem]{Proposition}
\newtheorem{assumption}{Assumption}
\theoremstyle{definition}\newtheorem{remark}{Remark}

\def\T{\mathbb{T}}
\def\m{\mathbf{m}}
\def\v{\mathbf{v}}
\def\x{\mathbf{x}}
\def\n{\mathbf{n}}
\def\R{\mathbb{R}}
\def\domain{I}
\def\C{C(\bar{\domain})}
\def\V{V_\ep}
\def\U{U}
\def\Lf{L_{f,\ep}}
\def\Lg{L_g}
\def\H{H(\domain)}
\def\L2{L^2(\domain)}
\def\gmax{g_M}
\def\Gmax{\Gamma_M}
\def\phione{\phi^{(1)}}
\def\czero{c^{(0)}}
\renewcommand{\div}[1]{\nabla\cdot\left(#1\right)}
\newcommand{\jump}[1]{\llbracket #1\rrbracket}
\newcommand{\intgrl}{\int\limits}
\newcommand{\ep}{\varepsilon}

\title[Multiphase modeling of tumor cords]{Multiphase modeling and qualitative analysis of the growth of tumor cords}
\author{Andrea Tosin}
\address{Department of Mathematics \\
		Politecnico di Torino \\
		Turin, Italy}
\email{andrea.tosin@polito.it}


\newcounter{savedenumi}

\begin{document}
\maketitle

\begin{abstract}
In this paper a macroscopic model of tumor cord growth is developed, relying
on the mathematical theory of deformable porous media. Tumor is modeled as a
saturated mixture of proliferating cells, extracellular fluid and extracellular
matrix, that occupies a spatial region close to a blood vessel whence cells get
the nutrient needed for their vital functions. Growth of tumor cells takes place
within a healthy host tissue, which is in turn modeled as a saturated mixture of
non-proliferating cells. Interactions between these two regions are accounted
for as an essential mechanism for the growth of the tumor mass. By weakening the
role of the extracellular matrix, which is regarded as a rigid non-remodeling
scaffold, a system of two partial differential equations is derived, describing
the evolution of the cell volume ratio coupled to the dynamics of the nutrient,
whose higher and lower concentration levels determine proliferation or death of
tumor cells, respectively. Numerical simulations of a reference two-dimensional
problem are shown and commented, and a qualitative mathematical analysis of some
of its key issues is proposed.
\end{abstract}

\section{Introduction}
Mathematical modeling of tumor growth dates back to many years, as the
considerable literature on the subject demonstrates. Recent surveys by Araujo
and McElwain \cite{MR2253816} and by Preziosi and Graziano \cite{MR2093306}
highlight the increasing interest by biologists and mathematicians toward the
qualitative and quantitative comprehension of this complex phenomenon, with the
common aim of devising suitable descriptive tools to predict the behavior of the
system and to virtually test the effect of different, possibly new therapies.

As a matter of fact, it is all but an easy task to conceive a mathematical
model capable to take into account the wide variety of biomechanical and
biochemical factors, in most part intimately correlated to each other, that
contribute to the overall dynamics of the system. This is particularly evident
if one considers that tumor growth is essentially a multiscale phenomenon, in
which subcellular, cellular, and tissue levels are strongly interconnected in
terms of cause and effect. In many cases, the reason for some macroscopic
outcome relies on the activation of specific internal mechanisms to the cells.
These usually trigger particular collective behaviors of groups of cells,
resulting finally in an observable effect at the macroscopic scale. In addition,
when dealing with living matter, as cells and tissues of human body are,
classical laws of Newtonian mechanics are not in principle the only rules that
the evolution of the system obeys to, due to some sort of intelligence that
makes cells able to self-organize according to possible changes in their state
or in the environment they live in.

Without claiming to be thorough, and referring instead to the above-cited papers
for a comprehensive review on the state of the art, we briefly outline here some
basic facts about mathematical modeling of tumor growth at the macroscopic
scale, in order to sketch the context which the present work fits in.

A first thread of models of tumor growth to be mentioned (see e.g., Byrne
\cite{MR2005055} and the main references listed therein) is based upon a
particular geometry of the tumor mass, the so-called spheroid, which
essentially consists in a spherically-shaped three-dimensional aggregate of
cells, primarily produced \emph{in vitro} for experimental purposes. In this
specific configuration, growth is addressed by focusing on the time evolution of
the external radius $R(t)$ of the spheroid. An integro-differential equation is
obtained by equating the time variation of the volume of the spheroid to the
overall cell proliferation, under the assumption that the latter is somehow
related to the concentration of a chemical, usually oxygen, which provides the
cells with the nutrient needed for their vital functions. This way it is
unnecessary to explicitly track the evolution of the cell density, and the model
is closed by simply adding a stand-alone reaction-diffusion equation for the
oxygen concentration, however posed in a domain which changes in time according
to the evolution of $R(t)$. The resulting mathematical problem is an
integro-differential free boundary problem, which has been proved to admit
solutions with suitable regularity properties and to predict, in some proper
ranges of the parameters, the evolution of the system toward a steady state
characterized by a finite steady growth size of the spheroid (Bueno \emph{et
al.} \cite{MR2150346}, Friedman and Reitich \cite{MR1684873}).

A slight variant of this framework entails the introduction of two regions in
the spheroid, namely an outer viable zone in which cells proliferate and a
central necrotic core in which cells starve and die due to an insufficient
delivery of oxygen from the periphery of the tumor. In this case, specific
reaction-diffusion equations for the nutrient concentration are set up in either
region, accounting for the fact that living cells take up oxygen from the
environment, while necrotic cells only receive a small amount of oxygen by means
of diffusion. In addition, an integro-differential equation is written for the
evolution of the outer proliferating shell, in the same spirit as the
corresponding equation of the previous model. The resulting problem describes
once again the evolution of the nutrient concentration across the spheroid, but
it is characterized by two free boundaries delimiting the necrotic core and the
external periphery of the spheroid, respectively. From the mathematical point of
view, the problem has been extensively studied by Cui and Friedman
\cite{MR1815805}, who addressed the existence of a stationary solution and the
convergence to it of the time dependent model under suitable assumptions on the
parameters.

An alternative class of mathematical models of tumor growth relies instead on
the theory of deformable porous media. In such a context, the tumor is regarded
as a mixture of various mutually interacting components, which obey standard
coupled mass and momentum balances. One of the main novelties with respect to
the above-described spheroid models is that now no particular geometry is in
principle required to write the equations, so that the global conformation of
the tumor is in turn an unknown of the problem which can be duly studied by the
model itself. However, the assumption of specific geometrical properties, like
e.g., spherical or axial symmetry, or even the reference to one-dimensional
settings are customary in the literature also in this case, in order to address
simplified explorative models.

Mixture-based models usually include one or more state variables specifically
devoted to track the evolution in time and space of the cell population. In the
early models of spheroids this was not a major issue since, as recalled above,
the growth of the tumor was described at an essentially geometrical level by
invoking a volume balance. Ambrosi and Preziosi \cite{MR1909425} pointed out
that this amounts to assuming a constant cell volume ratio within the tumor, so
that the growth of the spheroid would be essentially dictated by the necessity
to preserve such a constraint in spite of cell proliferation. Therefore, those
models can be viewed as particular approximations of a more general multiphase
setting. In many cases, the tumor is modeled as a biphasic mixture consisting of
cells and extracellular material. Frequently, the former are further divided
into several subpopulations of proliferating, quiescent, and necrotic cells,
with possible interchanges among the populations determined by the availability
of some nutrients that cells need in order to carry out their vital functions.
The extracellular material is commonly identified with extracellular fluid or
extracellular matrix (stroma). It fills the interstices among the cells, so that
no voids are left within the mixture and the latter can be regarded as
saturated. In some models of vascular tumors, the extracellular material also
includes a distribution of blood vessels mimicking the co-opted vasculature
under angiogenic effects (Breward \emph{et al.} \cite{BRa}).

In the theory of deformable porous media, the velocities of the constituents of
the mixture are determined from suitable momentum equations accounting for the
internal stress of each phase, as well as for the mutual external interactions
among different phases. Therefore, as usual in continuum mechanics, suitable
constitutive relations for the components of the tumor have to be specified, in
order to render at the tissue level the proper mechanical behavior. It is plain
that this poses the major modeling problems, because living tissues can be
assimilated to classical materials only on first approximation. For a detailed
study of the mechanics involved in tumor growth we refer the interested reader
to Ambrosi and Mollica \cite{MR1914120}. Here we simply remark that in some
particular geometries mass balance equations may be sufficient to determine, at
a purely kinematic level, the velocities of the components, at least under
suitably reinforced assumptions on the saturation of the mixture (Bertuzzi
\emph{et al.} \cite{MR2111919}).

A rigorous derivation of model equations for avascular tumor growth according
to the principles of the theory of mixtures can be found in Byrne and Preziosi
\cite{BYc} and in Breward \emph{et al.} \cite{MR1920584}. Following analogous
theoretical guidelines, in this paper we are concerned instead with vascular
tumor structures, that develop mainly in the axial direction along the wall of a
blood vessel. Due to the particular shape they take, they are named tumor
cords. For such tumors no particular form of angiogenesis is needed, as they
directly catch the nutrient by diffusion from the vessel which they grow
around. Mathematical modeling of tumor cord growth has already been addressed
by a number of papers by Bertuzzi and coworkers
\cite{MR2111919,MR2180715,MR2167659}. In particular, they use a mixture
theory approach supplemented by proper kinematic constraints that allow to
express the velocity of the components of the mixture without invoking any
momentum balance. In addition, they assume cylindrical symmetry of the cords
around the blood vessels, further neglecting axial variations of the state
variables or performing suitable averages in the axial direction when needed,
so as to reduce the domain of the problem to an annulus representing the
two-dimensional cross-section of the cord. They also add a sophisticated
kinematic constraint to describe the evolution of an inner viable core in
interaction with an outer necrotic ring, which are modeled as two distinct
zones by means of specific equations for each of them. Finally, they propose a
detailed mathematical analysis of both the stationary and the evolution
problems generated by their model. Conversely, we are interested in modeling
tumor cord growth starting from sufficiently general principles of the theory
of mixtures, so as to provide a model able in principle to face many different
physical and geometrical configurations. A basic (i.e., minimal) version of the
model is developed and studied here. For further extensions and improvements, as
well as for more applications to different systems, we refer to Preziosi and
Tosin \cite{PT}.

The paper is organized into six more sections that follow this Introduction.
In Sect. \ref{sect-model} the equations of the model and the proper boundary
conditions are derived, then in Sect. \ref{sect-nondim} the mathematical
problem is stated in non-dimensional form. In Sect. \ref{sect-numres} numerical
simulations are shown and commented, with the aim of testing the ability of the
model to reproduce some basic relevant features of the system. Sections
\ref{sect-statprob}, \ref{sect-proofs} address the qualitative mathematical
analysis of some key issues of the problem, and Sect. \ref{sect-conclusions}
finally draws conclusions and briefly sketches research perspectives.

\section{The mathematical model}
\label{sect-model}
In this section we propose a mathematical model for the description of the
growth of a tumor mass along a source of nutrient, working in the framework of
the theory of mixtures. In doing this, one of our main goals will be the
minimality of the model, meaning that we will constantly aim at taking into
account only those mechano-chemical mechanisms strictly relevant for an
essential representation of the macroscopic phenomenon. Of course, we are aware
that in this way many interesting aspects, whose inclusion would make the model
closer to physical reality and perhaps less academic, are left out of the study.
However, we believe that a minimal model along with a related qualitative
mathematical analysis are necessary starting points before tackling more
complicated situations, and that such a minimal model may be eventually used as
the core of future more accurate models.

\begin{figure}[t]
    \begin{center}
        \includegraphics[width=9cm]{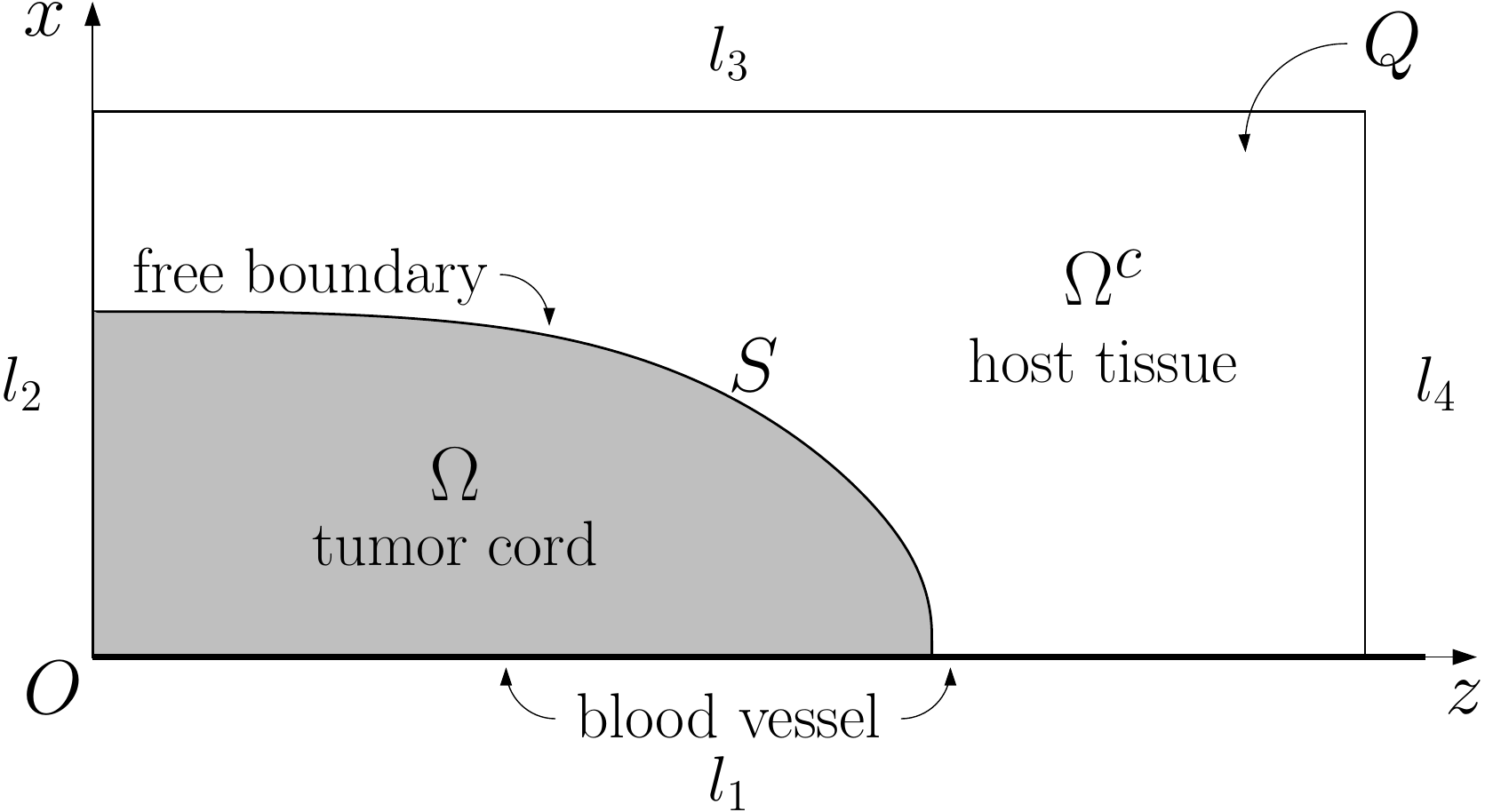}
    \end{center}
    \caption{The two-dimensional domain of the problem. A spatial region $Q$ is
        divided into two subregions $\Omega$ and $\Omega^c$, representing the
        tumor cord and the host tissue respectively, separated by a free
        boundary $S$. The blood vessel coincides with the horizontal $z$ axis.}
    \label{fig-domain}
\end{figure}

For the sake of definiteness, we will deduce the model having in mind the
two-dimensional geometry depicted in Fig. \ref{fig-domain}. Actually, since the
model is intrinsically multidimensional, any further extension to different
geometries and possibly to higher dimensions is mostly technical, requiring in
principle the same ideas up to a suitable reinterpretation of the symbolic form
of the differential equations that we will write.

Let us then consider a spatial region $Q\subset\R^2$ occupied by a mixture of
cells, extracellular fluid, and extracellular matrix (ECM), that we suppose
saturated. Denoting by $\phi_c$, $\phi_\ell$, $\phi_m$ the respective volume
ratios of the constituents, this condition is expressed by
\begin{equation}
    \phi_c+\phi_\ell+\phi_m=1 \qquad \text{in\ } Q.
    \label{saturation}
\end{equation}
The region $Q$ is internally divided into two subregions $\Omega$, $\Omega^c$
separated by a one-dimensional manifold $S$, so that $Q=\Omega\cup\Omega^c$ and
$\Omega\cap\Omega^c=S$. We assume that $\Omega$ represents the growing tumor
cord and $\Omega^c$ the surrounding host tissue, while $S$ plays in this context
the role of a free boundary delimiting the cord. As we will see, the model
guarantees that tumor cells located in $\Omega$ and normal cells located in
$\Omega^c$ never mix, hence the sole volume ratio $\phi_c$ is sufficient to
track the evolution of the whole cell population. The blood vessel around which
the tumor cord develops is schematized by the horizontal $z$ axis, and its
presence will be incorporated into the equations of the model in terms of
boundary conditions. This is possible, and is indeed useful, because we are not
interested in the interactions of the vessel with the tumor, so that its
geometry and mechanics can be duly ignored.

\subsection{Momentum equations}
It is assumed that normal cells and tumor cells differ only in that the latter
undergo an unregulated proliferation, while mechanical properties are
the same for both. Hence, one can write a unique momentum equation of the form
\begin{equation}
    -\div{\phi_c\T_c}+\phi_c\nabla{p}=\m_c \qquad \text{in\ } Q
    \label{momentum-cells}
\end{equation}
where:
\begin{itemize}
\item $p$ is introduced as a Lagrange multiplier to satisfy the saturation
constraint \eqref{saturation}, and is then interpreted as the interstitial
pressure of the extracellular fluid;

\item $\T_c$ is the excess stress tensor of the cellular matter, accounting for
cell-cell internal stress;

\item $\m_c$ is the resultant of the actions on the cellular matter due to the
interactions with the extracellular fluid and the extracellular matrix.
\end{itemize}
Notice that in Eq. \eqref{momentum-cells} inertia is neglected, in view of the
relatively slow dynamics of the cells during their growth process.

Assuming that cells behave like an elastic fluid, we introduce a scalar,
possibly nonlinear function $\Sigma=\Sigma(\phi_c)$ such that
\begin{equation}
    \T_c=-\Sigma(\phi_c)\mathbb{I},
    \label{sigma}
\end{equation}
$\mathbb{I}$ being the identity matrix. Positive values of $\Sigma$ for high
$\phi_c$ denote compression of the cells, while negative values for low $\phi_c$
imply a packing tendency. In addition, we postulate the existence of a value
$\phi_0>0$ such that $\Sigma(\phi_0)=0$, identifying a configuration in which
cells get in touch without exerting any stress on each other. In the sequel, we
will occasionally refer to $\phi_0$ as the stress-free cell volume ratio (see
Fig. \ref{fig-sigma_draft}).

\begin{figure}[t]
    \centering
    \includegraphics[width=6cm]{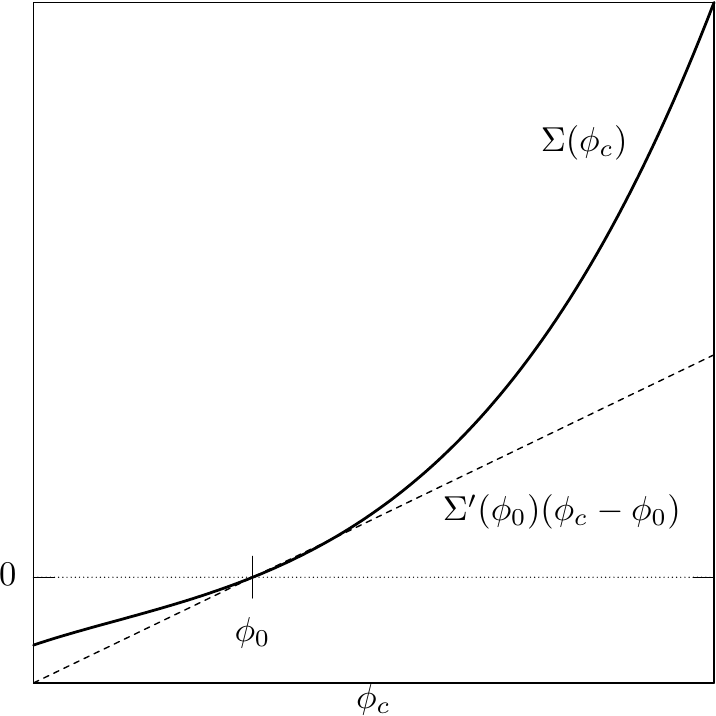}
    \caption{Cell stress function $\Sigma(\phi_c)$ and its first order
        approximation (tangent line) at the cell stress-free volume ratio
        $\phi_c=\phi_0$.}
    \label{fig-sigma_draft}
\end{figure}

Concerning the external actions, it is convenient to split $\m_c$ into two
contributions $\m_{c\ell}$, $\m_{cm}$ specifically related to the interactions
of the cells with the fluid and the ECM. More in detail, we can imagine a simple
viscous friction among the components of the mixture, that is
\begin{equation*}
    \m_{c\alpha}\propto \v_\alpha-\v_c, \qquad \alpha=\ell,\,m
\end{equation*}
where $\v_c,\,\v_\ell,\,\v_m$ denote the velocities of the constituents. In
particular, in the classical Darcy-like theory of porous media it is customary
to express the interaction with the fluid phase as
\begin{equation}
    \m_{c\ell}=\frac{\phi_\ell^2}{K}(\v_\ell-\v_c),
    \label{mcl}
\end{equation}
with $K>0$ representing the permeability of the mixture. Regarding the
interaction of the cells with the extracellular matrix, the theory is instead by
far less classical and consolidated, and certainly viscous friction can be
possibly accepted only as a first approximation. Indeed, cells are known to
attach to the ECM via suitable adhesion sites, and to detach only in presence
of sufficiently strong forces (Baumgartner \emph{et al.} \cite{BAa}, Canetta
\emph{et al.} \cite{CAa}, Sun \emph{et al.} \cite{SUa}), therefore such an
interaction is definitely not a pure sliding. However, to keep things simple we
set
\begin{equation}
    \m_{cm}=\Lambda_{cm}(\v_m-\v_c)
    \label{mcm}
\end{equation}
for a certain constant $\Lambda_{cm}>0$, referring the interested reader to
Preziosi and Tosin \cite{PT} for a more detailed mathematical modeling of this
term of the equations.

Introducing Eqs. \eqref{sigma}, \eqref{mcl}, \eqref{mcm} into Eq.
\eqref{momentum-cells} we find, after some standard manipulations,
\begin{equation}
    \nabla{(\phi_c\Sigma(\phi_c))}+\phi_c\nabla{p}=
        \frac{\phi_\ell^2}{K}(\v_\ell-\v_c)+\Lambda_{cm}(\v_m-\v_c)
        \qquad \text{in\ } Q.
    \label{momentum-cells-subst}
\end{equation}

For the extracellular fluid, an equation similar to Eq. \eqref{momentum-cells}
can be written:
\begin{equation}
    -\div{\phi_\ell\T_\ell}+\phi_\ell\nabla{p}=\m_\ell \qquad \text{in\ } Q
    \label{momentum-fluid}
\end{equation}
with $\m_\ell=\m_{\ell c}+\m_{\ell m}$, that is the sum of the interactions
with the cells and the ECM, respectively. However, due to the interpretation of
$p$ provided above, the excess stress tensor $\T_\ell$ is commonly neglected
and the pressure is regarded as the main internal stress of the fluid, so that
Eq. \eqref{momentum-fluid} actually simplifies as
\begin{equation}
    \phi_\ell\nabla{p}=\m_{\ell c}+\m_{\ell m} \qquad \text{in\ } Q.
    \label{momentum-fluid-simpl}
\end{equation}

Relying again on a Darcy-like framework, we think of the interaction terms as
proportional to the relative velocities of the interacting constituents, with
specifically
\begin{equation}
    \m_{\ell c}=-\m_{c\ell}=-\frac{\phi_\ell^2}{K}(\v_\ell-\v_c)
    \label{mlc}
\end{equation}
due to an action-reaction principle, and
\begin{equation}
    \m_{\ell m}\approx 0
    \label{mlm}
\end{equation}
to emphasize the fact that the interaction of the extracellular fluid with the
extracellular matrix is actually negligible with respect to the corresponding
one of the cells (cf. Eq. \eqref{mcm}). It might be inferred that a more
accurate modeling of the latter may allow for a better expression of the former.
However, such topics are beyond the scope of a minimal model of tumor growth and
will therefore not be addressed in the present work.

Putting Eqs. \eqref{mlc}, \eqref{mlm} into Eq. \eqref{momentum-fluid-simpl}
yields the relation
\begin{equation}
    \v_\ell-\v_c=-\frac{K}{\phi_\ell}\nabla{p},
    \label{Darcy}
\end{equation}
which represents the classical form of the well-known Darcy's law.

Finally, an equation similar to Eqs. \eqref{momentum-cells},
\eqref{momentum-fluid} holds in principle also for the ECM:
\begin{equation}
    -\div{\phi_m\T_m}+\phi_m\nabla{p}=\m_m \qquad \text{in\ } Q,
    \label{momentum-ECM}
\end{equation}
with $\m_m=\m_{mc}+\m_{m\ell}=-\m_{cm}-\m_{\ell m}$ in view of the
action-reaction principle. However, it requires to specify the excess stress
tensor $\T_m$, which appears to be a very complicated object due to the fibrous
nature of the extracellular matrix and to the continuous remodeling it undergoes
when cells move within it. Therefore, wishing to focus mainly on the growth of
the tumor cord in interaction with a source of nutrient, we choose to consider
the ECM as a rigid non-remodeling scaffold in which cells live and grow and
extracellular fluid flows. As a consequence, we replace Eq. \eqref{momentum-ECM}
by
\begin{equation}
    \v_m=0 \qquad \text{in\ } Q
    \label{vm0}
\end{equation}
and we simply observe that $\T_m$ plays now the role of a (tensor) Lagrangian
multiplier to satisfy the constraint \eqref{vm0}. In addition, we take the
volume ratio $\phi_m$ as a constant, say
\begin{equation}
    \phi_m=1-\phi_\ast, \qquad 0<\phi_\ast<1,
    \label{phim}
\end{equation}
so that the saturation constraint \eqref{saturation} specializes as
\begin{equation}
    \phi_c+\phi_\ell=\phi_\ast \qquad \text{in\ } Q.
    \label{saturation-reduced}
\end{equation}

It is interesting to distinguish in the momentum equation
\eqref{momentum-cells-subst} for the cells the contributions due to the
interaction with the extracellular fluid, represented by the terms
$\phi_c\nabla{p}$ and $\m_{\ell c}$, from the remaining ones. In most cases, one
can assume that the former are negligible, in terms of order of magnitude, with
respect to the latter, i.e.,
\begin{equation}
    \phi_c\vert\nabla{p}\vert,\ \vert\m_{\ell c}\vert
        = o\left(\vert\div{\phi_c\T_c}\vert,\ \vert\m_{mc}\vert\right)
        \qquad \text{in\ } Q,
    \label{sub-balance}
\end{equation}
so that the main balance is between the intercellular stress and the
interaction force with the extracellular matrix. Owing to this, Eq.
\eqref{momentum-cells-subst} along with Eq. \eqref{vm0} provides a particularly
simple expression for the velocity $\v_c$ of the cells:
\begin{equation}
    \v_c=-\frac{1}{\Lambda_{mc}}\nabla{(\phi_c\Sigma(\phi_c))}.
    \label{vc}
\end{equation}
By consequence, Eqs. \eqref{Darcy}, \eqref{saturation-reduced} give in turn
\begin{equation}
    \v_\ell=-\left[\frac{1}{\Lambda_{mc}}\nabla{(\phi_c\Sigma(\phi_c))}
        +\frac{K}{\phi_\ast-\phi_c}\nabla{p}\right].
    \label{vl}
\end{equation}

\subsection{Mass balance equations}
For the system constituted by cells and extracellular fluid, one can also write
a balance of mass under the assumption that they form a closed mixture, that is
a mixture which does not exchange matter with the external environment. In more
detail, we have to take into account that tumor cells, unlike normal cells of
the host tissue, are characterized by an intense proliferation, which is mainly
responsible for the macroscopic growth of the tumor cord. In order to focus
specifically on this phenomenon, we can reasonably disregard the standard
reproduction activity of the host cells, so that our mass balance equations
read
\begin{equation}
    \left.
    \begin{array}{rcl}
        \dfrac{\partial\phi_c}{\partial t}+\div{\phi_c\v_c} & = &
            G(\phi_c,\,c) \\
        \\[-0.2cm]
        \dfrac{\partial\phi_\ell}{\partial t}+\div{\phi_\ell\v_\ell} & = &
            -G(\phi_c,\,c)
    \end{array}
    \right\}
    \qquad \text{in\ } \Omega
    \label{mass-cord}
\end{equation}
\begin{equation}
    \left.
    \begin{array}{rcl}
        \dfrac{\partial\phi_c}{\partial t}+\div{\phi_c\v_c} & = & 0 \\
        \\[-0.2cm]
        \dfrac{\partial\phi_\ell}{\partial t}+\div{\phi_\ell\v_\ell} & = & 0
    \end{array}
    \right\}
    \qquad \text{in\ } \Omega^c.
    \label{mass-host}
\end{equation}
In Eqs. \eqref{mass-cord}, $G(\phi_c,\,c)$ represents a source/sink of cell
mass due to reproduction and death of tumor cells in connection with the
availability of some nutrient $c$. No contribution of ECM to the cell growth is
explicitly accounted for, though it is known that the extracellular matrix
contains growth factors. The reason is that, in the present context, ECM is not
supposed to play a primary role in the dynamics of the system, hence its
inclusion in the function $G$ would result at most in constant coefficients
related to the volume ratio $\phi_m$ (cf. Eq. \eqref{phim}).

It is assumed that dead tumor cells are degraded into extracellular fluid, and
conversely that the latter is consumed whenever a tumor cell duplicates to
originate new cells. We incidentally notice that Eqs. \eqref{mass-cord},
\eqref{mass-host} contain the implicit hypothesis that the density of the cells
and of the extracellular fluid is the same, and moreover that it is constant in
the whole mixture.

Denoting by $\chi_\Omega=\chi_\Omega(t,\,\x)$, $\x=(x,\,z)\in Q$, the
indicator function of the set $\Omega$, i.e.
\begin{equation*}
    \chi_\Omega(t,\,\x)=
    \begin{cases}
        1 & \text{if\ }\x\in\Omega\ \text{at time\ } t \\
        0 & \text{otherwise,}
    \end{cases}
\end{equation*}
we can substitute Eq. \eqref{vc} into the first of Eqs. \eqref{mass-cord} to get
a single equation for the cell volume ratio $\phi_c$:
\begin{equation}
    \frac{\partial\phi_c}{\partial t}-\div{\frac{\phi_c}{\Lambda_{mc}}
        \nabla{(\phi_c\Sigma(\phi_c))}}=
        G(\phi_c,\,c)\chi_\Omega \qquad \text{in\ } Q.
    \label{phic}
\end{equation}
The fluid volume ratio $\phi_\ell$ is then determined algebraically \emph{a
posteriori} from Eq. \eqref{saturation-reduced} as $\phi_\ell=\phi_\ast-\phi_c$.

In addition, summing term by term Eqs. \eqref{mass-cord}, \eqref{mass-host} and
recalling again Eq. \eqref{saturation-reduced} we further discover
\begin{equation*}
    \div{\phi_c\v_c+\phi_\ell\v_\ell}=0 \qquad \text{in\ } Q,
\end{equation*}
which, together with Eqs. \eqref{vc}, \eqref{vl}, implies
\begin{equation*}
    -\div{K\nabla{p}}=\div{\frac{\phi_\ast}{\Lambda_{mc}}
        \nabla{(\phi_c\Sigma(\phi_c))}}
        \qquad \text{in\ } Q.
\end{equation*}
This means that, in view of assumption \eqref{sub-balance}, we allow for a
nonconstant pressure field $p$ that can be in turn recovered \emph{a
posteriori} as a by-product of the integration of Eq. \eqref{phic}. Finally,
once also $p$ is known, Eq. \eqref{vl} yields the velocity $\v_\ell$ of the
extracellular fluid.

\subsection{Cell growth and nutrient diffusion}
The function $G=G(\phi_c,\,c)$ appearing in Eqs. \eqref{mass-cord},
\eqref{phic} describes gain (if $G>0$) or loss (if $G<0$) of cell mass caused
by reproduction or death of tumor cells, respectively. These are triggered by
the presence of some nutrient within the mixture, whose concentration per unit
volume is denoted by $c$. As a matter of fact, such a nutrient can be mainly
identified with the oxygen carried by the blood, which diffuses through the
mixture from the wall of the vessel around which tumor cells grow. It is assumed
that oxygen molecules occupy no space within the mixture.

It can be inferred that the basic mechanism by which the available amount of
oxygen affects the vital functions of the cells consists in the existence of a
critical threshold $c_0>0$ of the concentration $c$. When the latter is above
this value, cells are sufficiently fed and can duplicate. Conversely, when
$c$ falls below $c_0$, cells starve and die. A more sophisticated process may
include also a quiescency stage, taking place when oxygen concentration is below
$c_0$ but above a second threshold $c_1<c_0$. In this state, tumor cells neither
reproduce nor die, they simply enter a survival state waiting for being
reactivated, if $c$ grows above $c_0$, or for definitely dying, in case $c$
falls further below $c_1$.

>From the modeling point of view, it is convenient to express the function $G$
in the form
\begin{equation}
    G(\phi_c,\,c)=g(\phi_c)\Gamma(c),
    \label{G}
\end{equation}
so that the new function $g$ carries the information about the specific growth
dynamics of the cells, while the function $\Gamma$ accounts for growth
regulation by oxygen concentration.

\begin{figure}[t]
    \begin{center}
        \includegraphics[width=5.9cm]{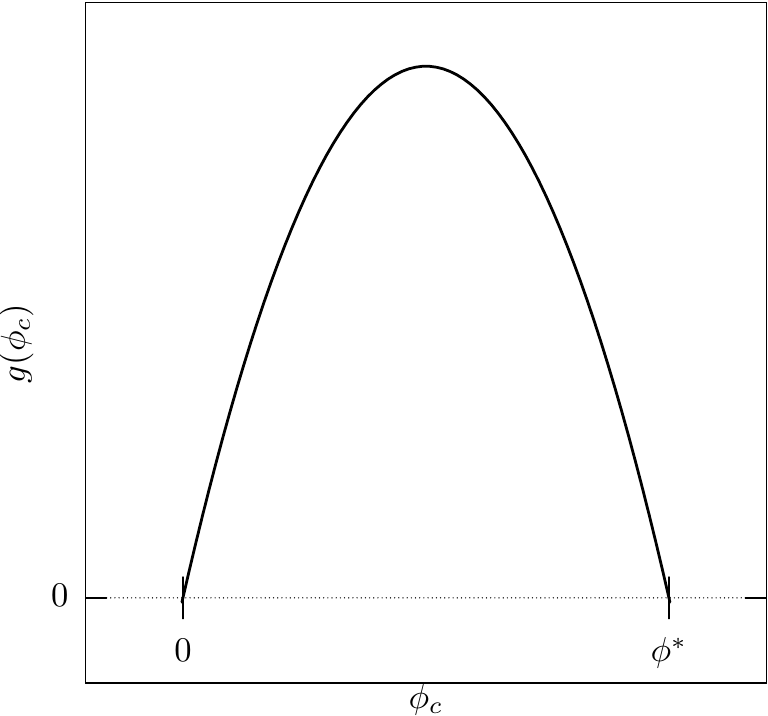}
        \qquad
        \includegraphics[width=5.9cm]{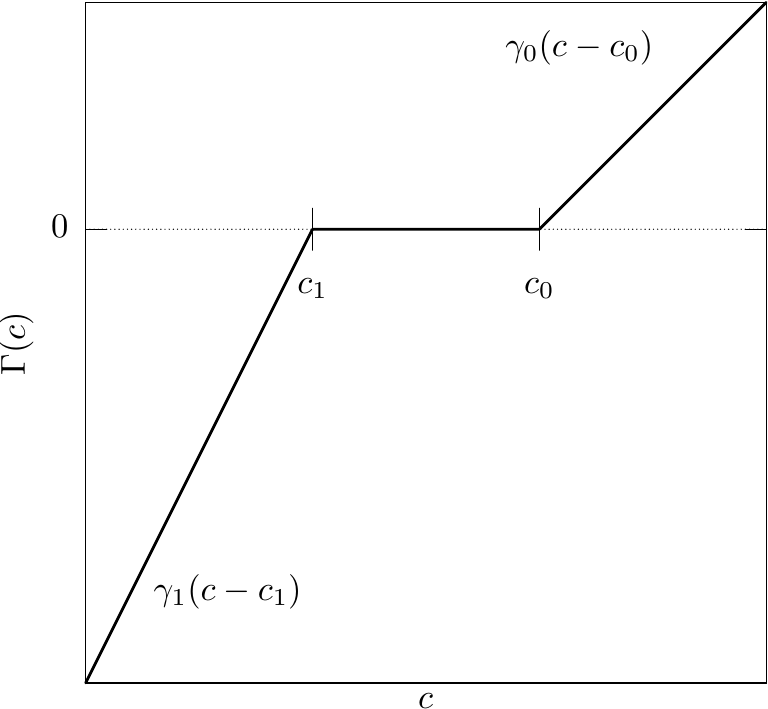}
    \end{center}
    \caption{Logistic cell growth function $g(\phi_c)$ (left) and growth
        regulation function $\Gamma(c)$ with two thresholds (right).}
    \label{fig-g_gamma}
\end{figure}

We want to point out that a major requirement on the cell volume ratio $\phi_c$
is that it should be bounded between $0$ and $\phi_\ast$, in order to guarantee
physical consistency of the solutions returned by the model (cf. Eq.
\eqref{saturation-reduced}). Cell growth mechanisms play undoubtedly a
fundamental role in this respect, and the choice of $g$ turns out to be crucial
in order to generate physically significant mathematical models. Therefore, it
should rely not only upon physical intuition but also on a qualitative knowledge
of the mathematical properties of the model itself. Specific features of $g$
leading to mathematically well-posed problems will be discussed in detail
in Sect. \ref{sect-statprob}. Here we simply anticipate that a logistic term
(see Fig. \ref{fig-g_gamma}, left)
\begin{equation*}
    g(\phi_c)=\phi_c(\phi_\ast-\phi_c)
\end{equation*}
is a possible, physically reasonable choice that indeed works.

Concerning $\Gamma$, it will be clear at the end of Sect. \ref{sect-statprob}
that it mainly influences the distribution of tumor cells far from the blood
vessel at the steady state. Specifically, it determines the maximum size reached
by the cord in the transverse $x$ direction. For the moment, we just observe
that the two scenarios previously discussed may correspond to
\begin{equation*}
    \Gamma(c)=\gamma(c-c_0), \qquad \gamma>0,
\end{equation*}
and to
\begin{equation*}
    \Gamma(c)=
        \begin{cases}
            \gamma_0(c-c_0) & \text{if\ } c\geq c_0 \\
            0 & \text{if\ } c_1<c<c_0, \\
            \gamma_1(c-c_1) & \text{if\ } c\leq c_1
        \end{cases}
    \qquad \gamma_0,\,\gamma_1>0,
\end{equation*}
respectively (see Fig. \ref{fig-g_gamma}, right).

As already mentioned, oxygen reaches cells by diffusion from the vessel wall.
Actually, to some extent it is also advected through the mixture by the
extracellular fluid, but it is known that advection is globally quite
inefficient with respect to diffusion, due to the relatively slow dynamics of
the constituents of the mixture compared to the high diffusivity of oxygen
molecules (Netti and Jain \cite{NJ}). Therefore we can write for $c$ a standard
reaction-diffusion equation of the form
\begin{equation}
    \frac{\partial c}{\partial t}-D\Delta{c}=-\alpha\phi_c c\chi_\Omega
        \qquad \text{in\ } Q,
    \label{diffusion}
\end{equation}
where $D>0$ is the oxygen diffusion coefficient within the mixture, and the term
at the right-hand side models the uptake of nutrient by the sole tumor cells
(this is the reason for the indicator function $\chi_\Omega$) as proportional to
the product $\phi_c c$ via a phenomenological constant $\alpha>0$.

\begin{remark}
It is worth noting that in the cell growth function $G$, as well as in the
right-hand side of the first of Eqs. \eqref{mass-host}, we have not included a
term related to the natural death of cells (e.g., a term of the form
$-\delta\phi_c$, $\delta>0$). Similarly, in Eq. \eqref{diffusion} neither oxygen
absorption by normal host cells nor possible chemical degradation of the
nutrient (e.g., by a term like $\tau^{-1}c$, with $\tau>0$ proportional to the
characteristic half-life of oxygen molecules) has been taken into account. Of
course, these effects could be straightforwardly incorporated in the model, but
in the present context they seem negligible with respect to the main dynamics we
are concerned with.
\end{remark}

\subsection{Boundary and interface conditions}
The system of equations \eqref{phic}, \eqref{diffusion} needs to be supplemented
by suitable conditions for the two main state variables of the problem, i.e.,
$\phi_c$ and $c$, at the boundary of the domain as well as on the interface $S$
between the regions $\Omega$ and $\Omega^c$.

With reference to Fig. \ref{fig-domain}, we see that the domain $Q$ has four
boundaries, whose just one is physical, namely the edge $l_1$ representing the
blood vessel, while the remaining three serve uniquely to confine geometrically
the domain of the problem. In particular, $l_2$ can be considered as a symmetry
boundary, in the sense that one can imagine a symmetric evolution of the system
for $z<0$, while the upper and right boundaries $l_3,\,l_4$ should be thought of
as `sufficiently far' (ideally, $x,\,z\to+\infty$) in the host tissue to be
unaffected by the dynamics of the growing tumor cord.

In the sequel, the symbol $\n$ will denote a generic outward normal unit
vector, which has to be referred from time to time to the appropriate boundary
or interface under consideration.

At the vessel wall we assume no detachment of both tumor and normal cells, which
essentially amounts to setting their normal velocity to zero. Moreover, we
prescribe the typical oxygen concentration $c_b$ found in the blood. Taking Eq.
\eqref{vc} into account, we have therefore the following conditions:
\begin{equation}
    \nabla{(\phi_c\Sigma(\phi_c))}\cdot\n=0, \quad c=c_b
        \qquad \text{on\ } l_1,
    \label{bc-vessel}
\end{equation}
that in our two-dimensional geometry specialize as
$\partial_x(\phi_c\Sigma(\phi_c))=0$, $c=c_b$ for $x=0$.

On the symmetry boundary we assign canonical symmetry conditions on $\phi_c$ and
$c$:
\begin{equation*}
    \nabla{\phi_c}\cdot\n=\nabla{c}\cdot\n=0
        \qquad \text{on\ } l_2,
\end{equation*}
which in our specific case take the form $\partial_z\phi_c=\partial_zc=0$ for
$z=0$.

Finally, at the `far' upper and right boundaries we assume that the
host tissue is unperturbed, i.e., unstressed, and we set zero flux of oxygen:
\begin{equation*}
    \phi_c=\phi_0, \quad \nabla{c}\cdot\n=0
        \qquad \text{on\ } l_3,\,l_4.
\end{equation*}
Note in particular that the second of these conditions is interpreted in the
present context as $\partial_xc=0$ on $l_3$, $\partial_zc=0$ on $l_4$.

According to the classical theory of deformable porous media (see e.g.,
Preziosi \cite{MR1422885}), at the interface $S$ between the tumor cord and
the host tissue suitable conditions on the continuity of the internal cell
stress and of the flux of oxygen have to be imposed:
\begin{equation}
    \jump{\phi_c\T_c\n}=0, \quad \jump{\nabla{c}}\cdot\n=0
        \qquad \text{on\ } S,
    \label{bc-interface}
\end{equation}
where $\jump{\cdot}$ denotes the jump across the interface $S$. In particular,
given the constitutive relation \eqref{sigma}, the condition on the cell stress
corresponds to $\jump{\phi_c\Sigma(\phi_c)}=0$, and, if one assumes that
$\Sigma$ is a continuous function of $\phi_c$, this further reduces to
$\jump{\phi_c}=0$, that is the continuity of the cell volume ratio across $S$.

In addition, the motion of $S$ as a material interface for the cellular matter
is regulated by the following ordinary differential equation:
\begin{equation*}
    \frac{d\x(t)}{dt}\cdot\n=\v_c(t,\,\x(t))\cdot\n
        \qquad \forall\,\x\in S,
\end{equation*}
where $\v_c$ is given by Eq. \eqref{vc}, which simply accounts for the fact
that, in a Lagrangian framework, the points of $S$ move with the normal velocity
of the cells themselves.

\section{Nondimensional statement of the problem}
\label{sect-nondim}
Let $L$, $\tau$, $\Phi$, $C$ denote characteristic values of length, time,
cell volume ratio, and nutrient concentration respectively, that we choose such
that
\begin{equation*}
    L=\sqrt{D\tau}, \quad \Phi=\phi_\ast, \quad C=c_b.
\end{equation*}
Let us moreover introduce the non-dimensional variables $\tilde\x$, $\tilde t$,
$\tilde\phi_c$, $\tilde c$ defined by the relations
\begin{equation*}
    \x=L\tilde\x, \quad t=\tau\tilde t, \quad \phi_c=\phi_\ast\tilde\phi_c,
        \quad c=c_b\tilde c,
\end{equation*}
the non-dimensional parameters
\begin{equation*}
    \tilde\alpha = \alpha\phi_\ast c_b, \qquad
        \tilde\phi_0=\frac{\phi_0}{\phi_\ast},
\end{equation*}
and the non-dimensional functions
\begin{equation*}
    \tilde\Sigma(\tilde\phi_c)=\frac{\phi_\ast}{\Lambda_{mc}D}
        \Sigma(\phi_\ast\tilde\phi_c), \qquad
    \tilde g(\tilde\phi_c)=\frac{1}{\phi_\ast}g(\phi_\ast\tilde\phi_c), \qquad
    \tilde\Gamma(\tilde c)=\tau\Gamma(c_b\tilde c).
\end{equation*}
Dropping from now on the tildes from the dimensionless variables and the
subscript `$c$' from $\phi_c$ to simplify the notation, the problem originated
by Eqs. \eqref{phic}, \eqref{diffusion}, along with Eq. \eqref{G} and the
boundary and interface conditions \eqref{bc-vessel}-\eqref{bc-interface}, can be
formulated in non-dimensional form as
\begin{equation}
    \begin{cases}
        \dfrac{\partial\phi}{\partial t}-\nabla\cdot
            [\phi\nabla(\phi\Sigma(\phi))]=g(\phi)\Gamma(c)\chi_{\Omega}
            & \text{in\ }  Q \\
        \\[-0.3cm]
        \dfrac{\partial c}{\partial t}-\Delta c=
            -\alpha\phi c\chi_{\Omega}
            & \text{in\ } Q \\
        \\[-0.3cm]
        \partial_{x}(\phi\Sigma(\phi))=0, \quad
             c=1 & \text{for\ } x=0 \\
        \\[-0.3cm]
        \partial_{ z}\phi=0, \quad \partial_{z} c=0 &
            \text{for\ } z=0 \\
        \\[-0.3cm]
        \phi=\phi_0, \quad \partial_{x} c=0 &
            \text{on\ } l_3 \\
        \\[-0.3cm]
        \phi=\phi_0, \quad \partial_{z} c=0 &
            \text{on\ } l_4 \\
        \\[-0.3cm]
        \jump{\phi}=0, \quad
            \jump{\nabla c}\cdot\n=0 &
            \text{on\ } S,
    \end{cases}
    \label{nondim_problem}
\end{equation}
plus the free boundary condition
\begin{equation}
    \frac{d\x(t)}{dt}\cdot\n=[-\nabla(\phi\Sigma(\phi))](t,\,\x(t))\cdot\n
        \qquad\forall\,\x\in S.
    \label{nondim_free_bound_cond}
\end{equation}

The choice of $\phi_\ast$, namely the per cent amount of free space left by the
extracellular matrix, as reference value for the cell volume ratio implies that
at a non-dimensional level we should expect $0\leq\phi\leq 1$ in $Q$ in view of
Eq. \eqref{saturation-reduced}, as well as $0<\phi_0<1$ coherently with the fact
that in the dimensional formulation of the model $\phi_0$ is implicitly assumed
to be positive and lower than $\phi_\ast$ for physical consistency reasons.
Similarly, since we expect $c_b$ to be the maximum value of nutrient
concentration found in the system, we consider $0\leq c\leq 1$ in $Q$.

\begin{figure}[t]
    \begin{center}
        $t=100$
        \includegraphics[width=11cm]{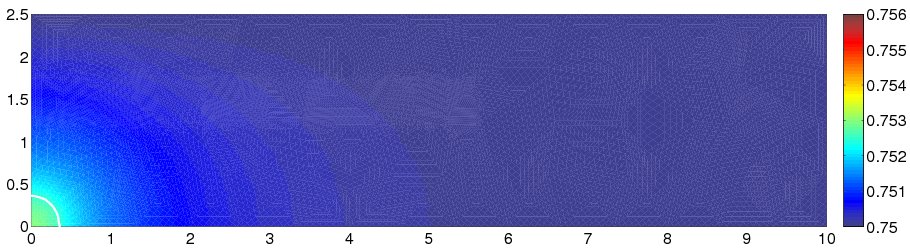}
        \\[0.5cm]
        $t=325$
        \includegraphics[width=11cm]{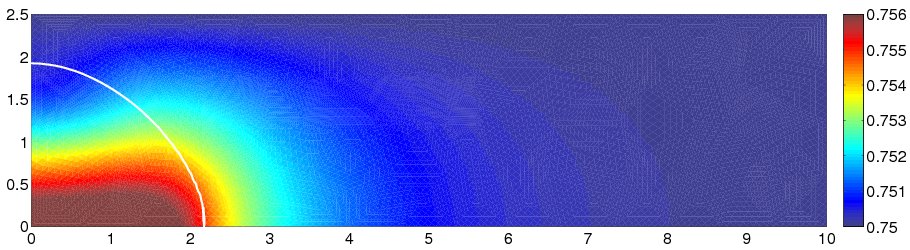}
        \\[0.5cm]
        $t=650$
        \includegraphics[width=11cm]{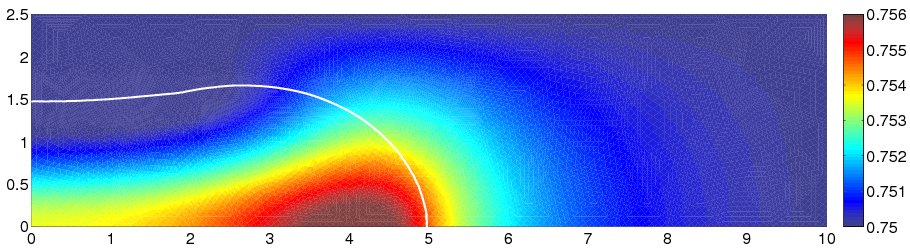}
        \\[0.5cm]
        $t=900$
        \includegraphics[width=11cm]{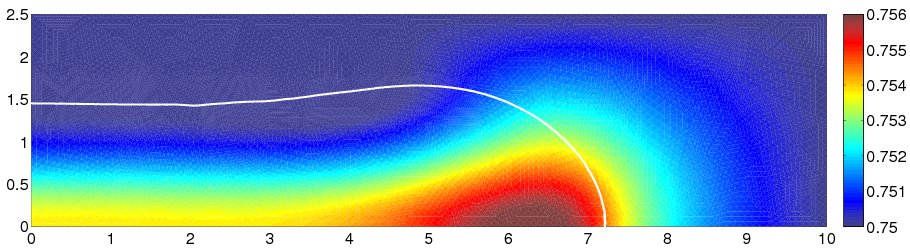}
    \end{center}
    \caption{Contour plot of the cell volume ratio $\phi$ at different
        times, with values increasing from blue to red (range from
        $0.75$ to $0.756$). White line indicates the position of the interface
        $S$.}
    \label{fig-density_2D}
\end{figure}

\section{Numerical results}
\label{sect-numres}
Equations \eqref{nondim_problem} have been numerically integrated over the
reference domain $Q=[0,\,2.5]\times [0,\,10]$ via a standard Finite Element
discretization, coupled with a level set technique to track the evolution of the
free boundary $S$. Specifically, the following functions have been used:
\begin{align}
    \Sigma(\phi) &= \frac{\mu}{\mu-1}\cdot
        \frac{\phi^{\mu-1}-\phi_0^{\mu-1}}{\phi}, \label{sigma-pormed}
    \\[0.3cm]
     g(\phi) &= \phi(1-\phi), \label{g} \\[0.3cm]
    \Gamma(c) &= \gamma(c-c_0), \label{Gamma}
\end{align}
along with this set of non-dimensional parameters:
\begin{equation}
    \mu=3,\qquad \phi_0=0.75,\qquad \gamma=0.7,\qquad c_0=0.8,\qquad
        \alpha=0.5.
    \label{nondim_param-numbers}
\end{equation}
We refer in particular the reader to Subsect. \ref{notex} for further
motivations on the choice of $\Sigma$.

\begin{figure}[t]
    \begin{center}
        $t=100$
        \includegraphics[width=11cm]{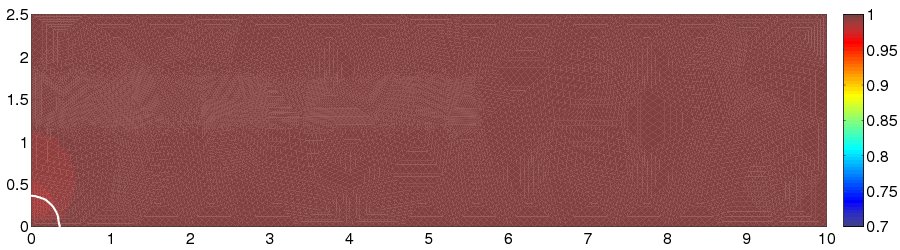}
        \\[0.5cm]
        $t=325$
        \includegraphics[width=11cm]{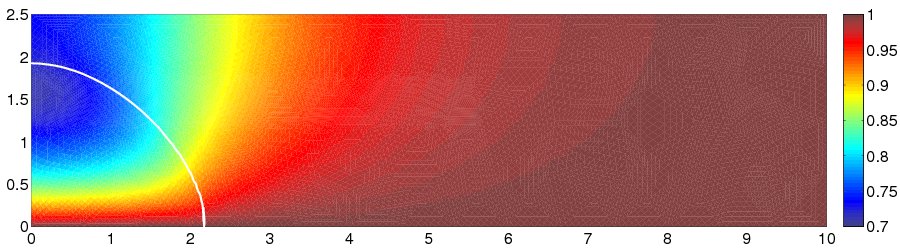}
        \\[0.5cm]
        $t=650$
        \includegraphics[width=11cm]{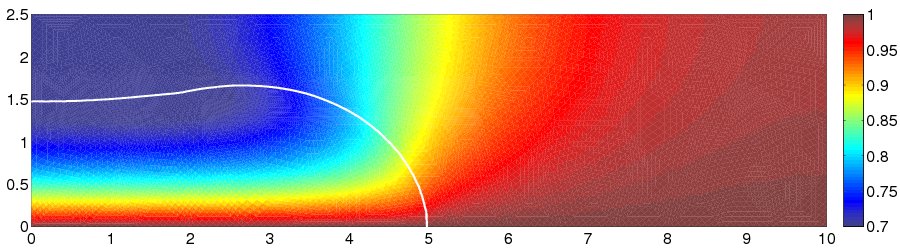}
        \\[0.5cm]
        $t=900$
        \includegraphics[width=11cm]{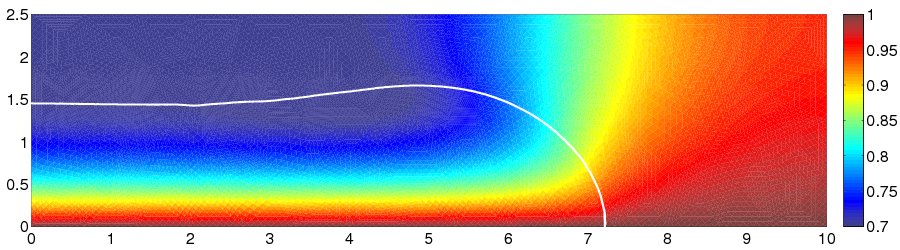}
    \end{center}
    \caption{Contour plot of the nutrient concentration $c$ at the
        corresponding times of the previous Figure \ref{fig-density_2D}, with
        values increasing from blue to red (range from $0.7$ to $1$). White line
        indicates the position of the interface $S$.}
    \label{fig-nutrient_2D}
\end{figure}

Figures \ref{fig-density_2D}, \ref{fig-nutrient_2D} show the evolution of the
cell volume ratio $\phi$ and the nutrient concentration $c$, respectively, at
different non-dimensional times listed on the left. The white line identifies in
each plot the position of the interface $S$, thus it visualizes the outer
boundary of the tumor cord. Initial conditions are set to $\phi(t=0)=\phi_0$,
$c(t=0)=1$ on the whole domain $Q$, with $\Omega(t=0)$ constituted by a quarter
of disk of small radius centered at the origin.

At early times ($t=100$) the cord begins to grow keeping a round homogeneous
shape, due to the high availability of nutrient across the domain. However, the
progressive consumption of nutrient by the cells makes oxygen concentration fall
below the critical threshold $c_0$ at the top of the cord ($t=325$), that is in
the farthest zone from the blood vessel. As a consequence, the transverse width
of the cord starts decreasing ($t=650$) and the whole structure takes a more
elongated shape. At later times ($t=900$), a front and a rear regions can be
clearly distinguished along the cord. The first one is a sort of growing rounded
head, where most living tumor cells are concentrated because they are globally
fed by a sufficient amount of nutrient. The second one is a kind of straight
tail of nearly constant transverse width, in which the distribution of cells and
oxygen depends uniquely on the distance from the blood vessel and has
substantially reached a steady state.

The dynamics of the system depicted by the model highlights two main issues
to be investigated, namely the axial and the transverse growth of the cord. In
more detail, the former is concerned with the determination of the shape and the
growth speed of the head in its motion along the blood vessel, possibly in
connection with traveling waves describing the evolution of the front
part of the interface $S$. The latter is instead aimed at characterizing the
typical transverse width of the tail, taking advantage of the above
consideration that the rear part of the cord is basically stationary with an
evident special dependence of $\phi$ and $c$ on the space variables $x$ and $z$.
In this paper, specifically in Sect. \ref{sect-statprob} below, we address this
second problem, leaving the first one for a possible forthcoming work.

We begin from a purely descriptive point of view, observing that in the tail of
the cord the cell volume ratio attains its maximum value at the vessel wall
($x=0$), as it can be expected for in that zone cell proliferation is constantly
sustained by a direct nutrient supply. Conversely, at the interface $S$ and in
the corresponding upper region inside the host tissue it decreases toward the
stress-free value $\phi_0$. Hence the expansion of the tail toward a transverse
stationary configuration causes tumor and host cells to achieve an equilibrium
spatial distribution characterized by a null mutual stress as well as a null
internal stress of the surrounding tissue. We incidentally notice that an
analogous unstressed state cannot be reached at the head of the cord, which is
never stationary for external cells are continuously compressed and pushed
away by growing tumor cells that penetrate into the host tissue.

With respect to oxygen distribution, we note that the tail can be ideally
divided into two stripes: An inner one, extending from the vessel to a certain
depth inside the cord, where $c\geq c_0$, and an outer one, reaching the
periphery of the cord, in which $c<c_0$. It is straightforward to assimilate
them to the so-called viable and necrotic regions, respectively, whose existence
is experimentally confirmed and explicitly described by means of suitable
equations in the mathematical model of tumor cords by Bertuzzi \emph{et al.}
\cite{MR2111919}. However, we want to point out that in the present context the
formation of these two distinct regions is not postulated \emph{a priori} as a
modeling assumption, but is recovered \emph{a posteriori} as a result of the
specific dynamics entailed by the model itself. Oxygen concentration, which is
maximum at its source, namely the blood vessel, decreases while diffusing
through the cord due to the absorption by tumor cells. The farther from the
vessel cells are the lower the quantity of nutrient they are reached by is, so
that some of them starve and can no longer proliferate. The process continues
until in a certain portion of the cord the net proliferation rate is zero: Then
that portion of cord is in equilibrium, it stops growing transversally and
finally reaches a steady state.

\section{The stationary problem}
\label{sect-statprob}
This part of the paper is devoted to the analysis of the mathematical model
presented in the previous Sects. \ref{sect-model} and \ref{sect-nondim}. In
particular, referring to the nondimensional form \eqref{nondim_problem} of the
equations, we consider a steady state in which the cord is supposed to be
infinitely long in the longitudinal $z$ direction, and to have grown up to a
width $w>0$ in the transverse $x$ direction (see Fig. \ref{fig-domain_1D}). In
such a configuration, the state variables $\phi$ and $c$ depend only on $x$ in
$\Omega$.

\begin{figure}[t]
    \begin{center}
        \includegraphics[width=9cm]{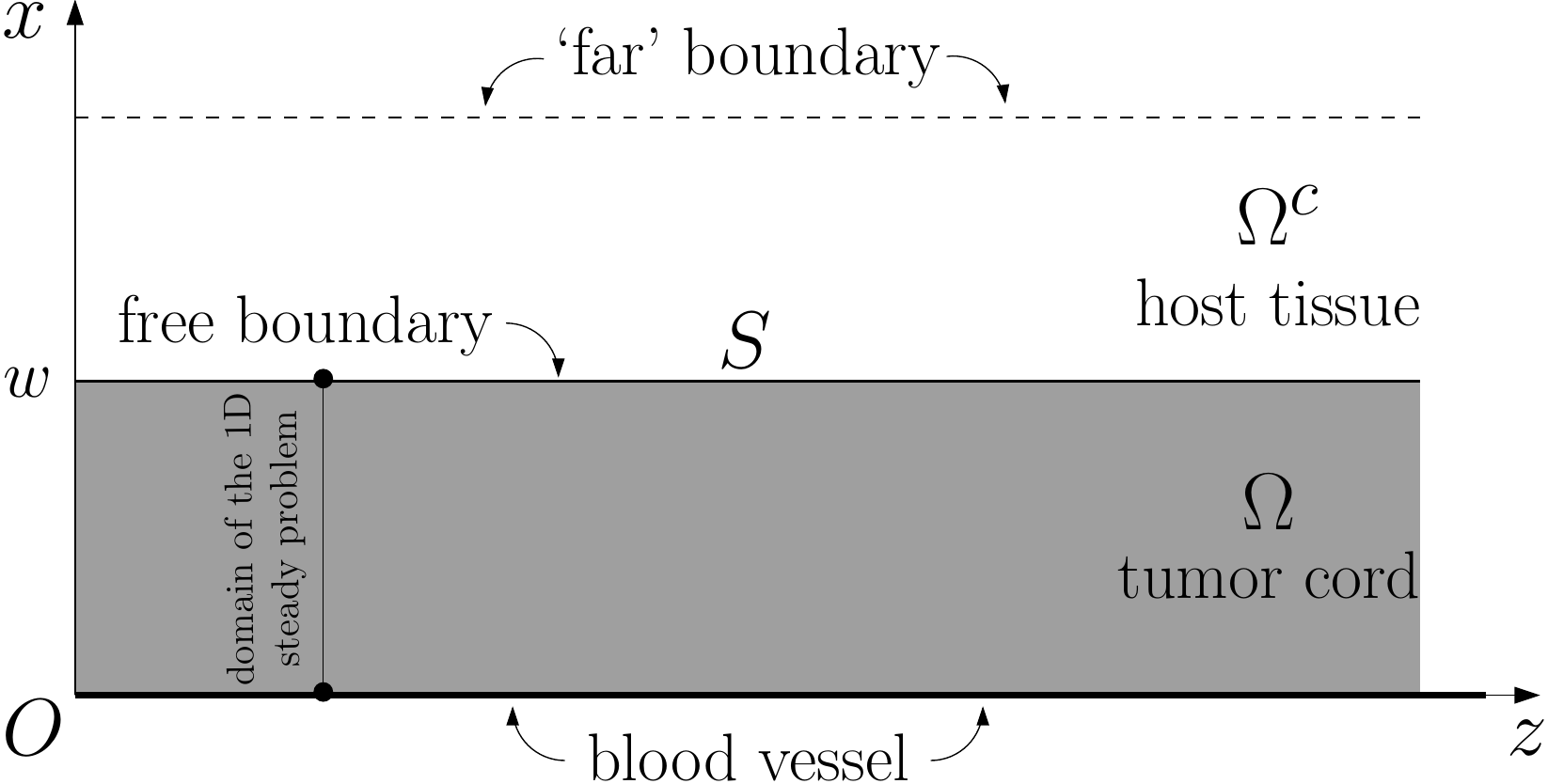}
    \end{center}
    \caption{Geometry used to study the one-dimensional stationary problem. The
        solution $(\phi,\,c)$ is assumed to be independent of $z$, which reduces
        the domain of the equations to the interval $[0,\,w]$ for the variable
        $x$.}
    \label{fig-domain_1D}
\end{figure}

In this geometrical setting, the free boundary $S$ coincides with the line
$x=w$, while the tumor cord and the host tissue are represented by the
stripe $0<x<w$ and the half-plane $x>w$, respectively. According to Eqs.
\eqref{nondim_problem}, an admissible steady solution $(\phi,\,c)$ in $\Omega^c$
satisfies $\phi=\phi_0$, $\nabla{c}=0$, hence we can confine ourselves to the
region $\Omega$ assuming the continuity conditions $\phi=\phi_0$,
$\partial_xc=0$ for $x=w$.

The objectives of the present study are twofold:

\begin{enumerate}
\item [(a)] As a first step (cf. Subsect. \ref{exreg}), to establish existence
of physically significant solutions with appropriate regularity. In this stage,
the width $w$ of the cord is regarded to all purposes as a parameter of the
model.

\item [(b)] As a second step (cf. Subsect. \ref{solfreebound}), to solve, via a
perturbative approach compatible with the results obtained in the previous
point, the free boundary problem by identifying a steady growth value of the
cord width.
\end{enumerate}

Specifically, we consider the following class of boundary value
problems\footnote{The condition $\partial_x\phi=0$ for $x=0$ of Problem
\eqref{bvp} is equivalent to the original boundary condition
$\partial_x(\phi\Sigma(\phi))=0$ of Problem \eqref{nondim_problem} in a sense
that will be made precise later in Remark \ref{remark-bc}.}:
\begin{equation}
    \left\{
    \begin{array}{ll}
        -\partial_x^2F(\phi)=w^2g(\phi)\Gamma(c) & \text{in\ } \domain \\
        \\[-0.3cm]
        -\partial_x^2c+\alpha w^2\phi c=0 & \text{in\ } \domain \\
        \\[-0.3cm]
        \partial_x\phi=0,\ c=1 & x=0 \\
        \\[-0.3cm]
        \phi=\phi_0,\ \partial_xc=0 & x=1
    \end{array}
    \right.
    \label{bvp}
\end{equation}
in the unknowns $\phi,\,c:\bar\domain\to\R_+$, where:

\begin{enumerate}
\item [(i)] $\domain=(0,\,1)$ is the rescaled domain after the substitution
$x\to wx$, which reduces the free boundary problem to a fixed boundary
problem. Notice that $w$ appears now among the coefficients of the
equations.

\item [(ii)] $F=F(\phi)$ is a `generalized' stress function linked to the cell
stress function $\Sigma=\Sigma(\phi)$ by the relation
\begin{equation}
    F'(\phi)=\phi{(\phi\Sigma(\phi))}',
    \label{F-Sigma}
\end{equation}
where the superscript $'$ stands for derivation with respect to $\phi$. This
way it results
\begin{equation*}
    -\Delta F(\phi)=-\nabla\cdot[\phi\nabla(\phi\Sigma(\phi))]
\end{equation*}
in accordance with the first of Eqs. \eqref{nondim_problem}.

\item [(iii)] $g=g(\phi)$ is the specific cell growth function, which
determines how cells proliferate on the basis of their current distribution.

\item [(iv)] $\Gamma=\Gamma(c)$ is the growth regulation function, which
expresses promotion or inhibition of cellular proliferation as a consequence of
the availability of nutrient $c$.

\item [(v)] $\alpha>0$, $\phi_0\in (0,\,1)$ are phenomenological parameters
related to the consumption rate of nutrient by the cells and to the stress-free
cell density, respectively.
\end{enumerate}
In particular, concerning the functions $F,\,g,\,\Gamma$ we formulate the
following assumptions:

\begin{assumption}
\label{assF}
We assume $F\in C^1([0,\,1])$ with
$$ F'(\xi)\geq 0,\ \forall\,\xi\in[0,\,1], \qquad
    F'(\xi)>0,\ \forall\,\xi\in(0,\,1], $$
and in addition $F(0)=0$, $F(1)=1$.
\end{assumption}
\begin{remark}
Given the dimensionless cell stress function $\Sigma$, condition $F(0)=0$ is
readily obtained by properly choosing the integration constant in Eq.
\eqref{F-Sigma}. Conversely, we observe that the fulfillment of condition
$F(1)=1$ may simply require a suitable rescaling of $\Sigma$.
\end{remark}

\begin{assumption}
\label{assg}
We assume that $g$ is Lipschitz continuous on $[0,\,1]$, with Lipschitz
constant $\Lg>0$:
$$ \vert g(\xi_2)-g(\xi_1)\vert\leq \Lg\vert\xi_2-\xi_1\vert,
    \qquad \forall\,\xi_1,\,\xi_2\in[0,\,1]. $$
Moreover, we require
$$  g(\xi)\geq 0\ \text{for\ } \xi\in[0,\,1], \qquad
    g(\xi)<0\ \text{otherwise}. $$
\end{assumption}

\begin{assumption}
\label{assGamma}
We assume that $\Gamma$ is Lipschitz continuous on $[0,\,1]$, with Lipschitz
constant $L_\Gamma>0$:
$$ \vert\Gamma(\xi_2)-\Gamma(\xi_1)\vert\leq L_\Gamma\vert\xi_2-\xi_1\vert,
    \qquad \forall\,\xi_1,\,\xi_2\in[0,\,1]. $$
\end{assumption}

From Assumption \ref{assF} it follows that $F$ is monotonically increasing
and invertible on $[0,\,1]$, and that its inverse function, henceforth denoted
by $f$ for brevity, is in turn increasing and differentiable on every compact
subset of $(0,\,1]$. We agree to denote by $\Lf$ the Lipschitz constant of $f$
on every interval of the form $[F(\ep),\,1]$, $0<\ep<1$. Notice that
the assumptions on $F$ amount mainly to a request of ellipticity of the
nonlinear equation for $\phi$, which is however potentially degenerate at
$\phi=0$.

On the other hand, Assumptions \ref{assg}, \ref{assGamma} entail the boundedness
of $g,\,\Gamma$ on $[0,\,1]$, hence there exist constants $\gmax,\,\Gmax>0$ such
that $\vert g(\xi)\vert\leq \gmax$, $\vert\Gamma(\xi)\vert\leq\Gmax$ for all
$\xi\in[0,\,1]$.

For modeling reasons discussed in Sect. \ref{sect-nondim}, we are interested in
nonnegative continuous solutions $\phi,\,c$ bounded from above by $1$. In more
detail, in view of the possible degeneracy of the equation for $\phi$ in zero,
we require that $\phi$ be strictly positive in the domain $\bar\domain$,
therefore in the sequel it will be instrumental to refer to the following sets:
\begin{align*}
    \V &= \left\{\phi\in\C\,:\,\ep\leq\phi(x)\leq 1,\
        \forall\,x\in\bar\domain\right\}, \\
    \\[-0.3cm]
    \U &= \left\{c\in\C\,:\,0\leq c(x)\leq 1,\ \forall\,x\in\bar\domain\right\},
\end{align*}
where $\ep>0$ is a fixed parameter. Notice that we cannot expect solutions
$\phi>\phi_0$ on $\bar\domain$, for they would not match the boundary
condition of Problem \eqref{bvp} at $x=1$, hence we assume $\ep<\phi_0$.

\begin{remark}
\label{remark-epsilon}
We anticipate that, as far as the solution to the free boundary problem is
concerned, the specific value of the parameter $\ep$ in the range
$(0,\,\phi_0)$ is irrelevant, because the cell density $\phi$ can be shown to
satisfy automatically $\phi\geq\phi_0$ on the whole interval $\bar\domain$ (cf.
Theorem \ref{phi-phi0-freebound}). The introduction of $\ep$ is due to the
necessity to obtain existence in $\V\times\U$ of solutions to Problem
\eqref{bvp}, among which to look later for possible solutions to the free
boundary problem. However, the theory we are going to develop will provide as
by-product an optimal criterion on whose basis to select the value of
$\ep$.
\end{remark}

\begin{remark}
\label{remark-bc}
The boundary condition for $\phi$ in $x=0$ stated by model
\eqref{nondim_problem} can be rewritten, according to Eq. \eqref{F-Sigma}, as
$$ 0=\partial_x(\phi\Sigma(\phi))=\frac{F'(\phi)}{\phi}\partial_x\phi. $$
Since for $\phi\in\V$ one has $\phi\geq\ep>0$, thus $F'(\phi)>0$ in view
of Assumption \ref{assF}, we deduce
$$ \partial_x(\phi\Sigma(\phi))=0 \quad \text{for\ } x=0
    \quad \ \Longleftrightarrow \quad \
    \partial_x\phi=0 \quad \text{for\ } x=0, $$
which gives the equivalence of boundary conditions between Problems
\eqref{nondim_problem} and \eqref{bvp}, at least for functions belonging to the
class $\V$.
\end{remark}

In the following, we will consider the main results of our analysis without
going into the details of their proofs, so as to give an overall picture of the
mathematical problem and of its solution. The interested reader can find all
technical proofs postponed in Sect. \ref{sect-proofs}. We simply point out here
that the constant $C_P$, occasionally appearing in several forthcoming
formulas, is the Poincar\'e constant of the domain $\domain$.

\subsection{Existence and regularity of the solutions}
\label{exreg}
We begin by addressing the question of existence and regularity of solutions to
Problem \eqref{bvp} for fixed positive $w$. Under suitable assumptions on the
parameters of the equations, we will identify a family of solutions
$(\phi,\,c)\in\V\times\U$ attached to each $w$ ranging in an appropriate set of
values.

Fix then $w>0$. The strategy we adopt to solve Problem \eqref{bvp} consists of
the following steps:

\begin{enumerate}
\item We choose any $\varphi\in\V$ and study the boundary value problem
\begin{equation}
    \left\{
    \begin{array}{rcll}
        -\partial_x^2c+\alpha w^2\varphi c & = & 0 & \text{in\ } \domain \\
        \\[-0.3cm]
        c & = & 1 & x=0 \\
        \\[-0.3cm]
        \partial_xc & = & 0 & x=1,
    \end{array}
    \right.
    \label{bvp-c}
\end{equation}
establishing existence and uniqueness of a solution $c\in\U$.

\item We then choose any $\sigma\in U$ and study the boundary value problem
\begin{equation}
    \left\{
    \begin{array}{rcll}
        -\partial_x^2F(\phi) & = & w^2g(\phi)\Gamma(\sigma) &
            \text{in\ } \domain \\
        \\[-0.3cm]
        \partial_x\phi & = & 0 & x=0 \\
        \\[-0.3cm]
        \phi & = & \phi_0 & x=1,
    \end{array}
    \right.
    \label{bvp-phi}
\end{equation}
finding conditions on the parameters that guarantee existence and uniqueness of
a solution $\phi\in\V$.

\item As a consequence of the previous steps, we can define an operator
$A:\V\to\V$ that to every function $\varphi\in\V$ associates the solution
$\phi\in\V$ of Problem \eqref{bvp-phi} through the solution $c\in\U$ of Problem
\eqref{bvp-c}. By showing that $A$ satisfies the hypotheses of Schauder Fixed
Point Theorem, we finally find a function $\phi\in\V$, and consequently also a
function $c\in\U$, solving Problem \eqref{bvp} as desired.

\item At last, we prove that under the same conditions on the parameters
established in the previous steps it is possible to control the distance
between any solution $\phi\in\V$ to Problem \eqref{bvp} and $\phi_0$, so that a
perturbative expansion of $\phi$ about $\phi_0$ be justifiable in addressing
next the free boundary problem.
\end{enumerate}

\subsubsection*{Steps 1, 2: Well-posedness of Problems \eqref{bvp-c},
\eqref{bvp-phi}}
First, we state the main properties of the functions $c$, $\phi$ as they result
from any admissible data $\varphi$, $\sigma$ on the basis of Problems
\eqref{bvp-c}, \eqref{bvp-phi}.

\begin{proposition}
\label{wellpos-c}
To every $\varphi\in\V$ there exists a unique solution $c\in\U$ of Problem
\eqref{bvp-c}, which is further twice continuously differentiable and
monotonically decreasing on $\bar\domain$, with a maximum value equal to $1$
attained for $x=0$.
\end{proposition}

\begin{proposition}
\label{wellpos-phi}
Set
\begin{equation}
    \beta_1=\beta_1(\ep):=
        \sqrt{\frac{C_P\gmax\Gmax}{F(\phi_0)-F(\ep)}}, \qquad
    \beta_2=\beta_2(\ep):=C_P\sqrt{\Lg\Gmax\Lf}
    \label{beta12}
\end{equation}
and define
\begin{equation}
    \beta=\beta(\ep):=\max{\left(\beta_1(\ep),\,
        \beta_2(\ep)\right)}.
    \label{beta}
\end{equation}
If $\beta w<1$, then for every $\sigma\in\U$ Problem \eqref{bvp-phi} admits a
unique solution $\phi\in\V$.
\end{proposition}

>From the proof of Proposition \ref{wellpos-phi} (cf. Sect. \ref{sect-proofs}),
it turns out that the constants $\beta_1(\ep)$, $\beta_2(\ep)$
determine the existence and the uniqueness, respectively, of the solution
$\phi\in\V$ to Problem \eqref{bvp-phi}. Notice that both of them tend to blow
when $\ep$ approaches its limiting values. Indeed, Eqs. \eqref{beta12} show
on the one hand that $\beta_1\to+\infty$ when $\ep\to\phi_0^-$, and they
may imply on the other hand $\beta_2\to+\infty$ for $\ep\to 0^+$ due to the
possible degeneracy of Problem \eqref{bvp-phi} for $\phi=0$. In fact, observing
that the Lipschitz constant $\Lf$ of $f$ on $[F(\ep),\,1]$ can be
characterized as
\begin{equation}
    \Lf=\max_{\xi\in[F(\ep),\,1]}f'(\xi)=
        \max_{\xi\in[F(\ep),\,1]}\frac{1}{F'(f(\xi))}=
        \frac{1}{\displaystyle{\min_{\eta\in[\ep,\,1]}}F'(\eta)},
    \label{Lf}
\end{equation}
we deduce that if $F'(0)=0$ then $\min_{\eta\in[\ep,\,1]}F'(\eta)\to 0$
for $\ep\to 0^+$, thus $\Lf\to+\infty$.

\subsubsection*{Step 3: Solution to Problem \eqref{bvp}}
We will henceforth always assume $\beta w<1$, even when not explicitly stated.
Let us introduce the operator $A_1:\V\to\U$ that to every $\varphi\in\V$
associates the unique solution $c=A_1(\varphi)\in\U$ of Problem \eqref{bvp-c}.
Analogously, let $A_2:\U\to\V$ be the operator that to such $c\in\U$ associates
the unique solution $\phi=A_2(c)\in\V$ of Problem \eqref{bvp-phi} with
$\sigma=c$. By composition, we define the operator $A=A_2\circ A_1:\V\to\V$ that
to every $\varphi\in\V$ associates the unique solution $\phi\in\V$ of the
problem
\begin{equation*}
    \left\{
    \begin{array}{ll}
        -\partial_x^2F(\phi)=w^2g(\phi)\Gamma(c) & \text{in\ } \domain \\
        \\[-0.3cm]
        -\partial_x^2c+\alpha w^2\varphi c=0 & \text{in\ } \domain \\
        \\[-0.3cm]
        \partial_x\phi=0,\ c=1 & x=0 \\
        \\[-0.3cm]
        \phi=\phi_0,\ \partial_xc=0 & x=1.
    \end{array}
    \right.
\end{equation*}
Clearly, if $A$ has a fixed point in $\V$, that is if there exists a
function $\phi\in\V$ such that $A(\phi)=\phi$, then the pair
$(\phi,\,c=A_1(\phi))\in\V\times\U$ is a solution to Problem \eqref{bvp}. Our
task is therefore to show that $A$ admits such a fixed point.

For this, it is first useful to know that

\begin{proposition} \label{lipschitz-A12}
The operators $A_1:\V\to\U$, $A_2:\U\to\V$ are Lipschitz continuous, and so is
$A:\V\to\V$.
\end{proposition}

Moreover, since in applying fixed point techniques a key feature is usually
the compactness of the involved operator, the following result is particularly
welcome.

\begin{proposition} \label{compact-A12}
The operator $A_1:\V\to\U$ is compact, and so is $A:\V\to\V$.
\end{proposition}

Thanks to Propositions \ref{lipschitz-A12}, \ref{compact-A12}, we can apply
Schauder Fixed Point Theorem to solve Problem \eqref{bvp}. For the sake of
completeness, we explicitly recall in the statement of the theorem all
hypotheses that bring to the result.

\begin{theorem} \label{solution}
Let Assumptions \ref{assF}, \ref{assg}, \ref{assGamma} hold with $\beta w<1$,
where $\beta=\beta(\ep)$ is the constant defined by Eq. \eqref{beta}.
Then there exists a solution $(\phi,\,c)\in\V\times\U$ to Problem \eqref{bvp}
satisfying the a priori estimate
$$ \|1-c\|_\infty\leq\alpha w^2\|\phi\|_{\L2}. $$
\end{theorem}

Theorem \ref{solution} gives existence but not uniqueness of solutions
$(\phi,\,c)\in\V\times\U$ to Problem \eqref{bvp}. Observe however that, owing to
Proposition \ref{wellpos-c}, to \emph{any} fixed point $\phi\in\V$ of the
operator $A$ there corresponds a \emph{unique} $c=A_1(\phi)\in\U$ such that the
pair $(\phi,\,c)\in\V\times\U$ solves Problem \eqref{bvp}.

\subsubsection*{Step 4: Controlling the distance between $\phi$ and $\phi_0$}
For any solution $(\phi,\,c)\in\V\times\U$ to Problem \eqref{bvp}, we are
interested now in controlling the distance between $\phi$ and the constant
$\phi_0$, in view of the application of a perturbative technique to solve the
free boundary problem. We find out that the quantity $\beta_2w$ plays a
fundamental role, in fact we have:
\begin{theorem}
\label{distance}
Let $\phi\in\V$ be any solution to Problem \eqref{bvp}. Then
$$ \|\phi-\phi_0\|_\infty\leq\frac{\gmax}{\Lg C_P}(\beta_2w)^2. $$
\end{theorem}

>From this result we infer that, since $\beta_2 w\leq\beta w$, for even
moderately small values of $\beta w$ (with respect to $1$) the solution $\phi$
becomes close to $\phi_0$ provided $\gmax$ (maximum cell proliferation) and
$\Lg$ (maximum proliferation rate) satisfy $\gmax/\Lg=O(C_P)$, $C_P\simeq
0.64$ (cf. Sect. \ref{sect-proofs}).

\subsection{Solution to the free boundary problem}
\label{solfreebound}
So far we have obtained existence and regularity of solutions to Problem
\eqref{bvp} under a suitable condition on the admissible steady values of the
cord width $w$. However, we are mainly interested in finding stationary
solutions satisfying in addition the free boundary condition
\eqref{nondim_free_bound_cond}, that in the present setting rewrites as
$$ \partial_x(\phi\Sigma(\phi))=0 \qquad \text{for\ } x=1 $$
or, on the basis of the same considerations proposed in Remark \ref{remark-bc},
more explicitly as
\begin{equation}
    \partial_x\phi=0 \qquad \text{for\ } x=1.
    \label{bc-freebound}
\end{equation}
To this end, we formulate additional assumptions on $\Gamma$:
\begin{assumption}  \label{assGamma-2}
In addition to Assumption \ref{assGamma}, we assume that $\Gamma$ is
nondecreasing on $[0,\,1]$:
$$ \Gamma(\xi_1)\leq\Gamma(\xi_2), \quad \forall\,\xi_1,\,\xi_2\in[0,\,1],\
    \xi_1\leq\xi_2 $$
with $\Gamma(0)<0<\Gamma(1)$.
\end{assumption}
Notice that Assumptions \ref{assGamma} and \ref{assGamma-2} imply the existence
of a point $c_0\in (0,\,1)$ such that $\Gamma(c_0)=0$, with moreover
$\Gamma(\xi)\leq 0$ for $\xi\in[0,\,c_0)$, $\Gamma(\xi)>0$ for
$\xi\in(c_0,\,1]$.

We state first of all a crucial result that clarifies the role of the parameter
$\ep$.
\begin{theorem} \label{phi-phi0-freebound}
Let Assumption \ref{assGamma-2} hold and assume that a solution
$(\phi,\,c)\in\V\times\U$ exists to Problem \eqref{bvp}, such that $\phi$
fulfills the free boundary condition \eqref{bc-freebound}. Then
$\phi(x)\geq\phi_0$ for all $x\in\bar\domain$.
\end{theorem}
In view of Theorem \ref{phi-phi0-freebound}, the specific value of
$\ep\in(0,\,\phi_0)$ is not relevant to solve the free boundary problem,
indeed any possible solution, provided it exists, is `naturally' bounded from
below by $\phi_0$. Considering that Theorem \ref{solution} guarantees existence
of solutions if $\beta w<1$ and that, on the other hand, it results
$\beta\to+\infty$ when both $\ep\to 0^+$ and $\ep\to\phi_0^-$ (cf.
Proposition \ref{wellpos-phi}), we deduce that the optimal criterion for the
choice of $\ep$ aims at minimizing $\beta(\ep)$, so as to get the
widest possible range of admissible values of $w$.

Let us look for a first order expansion of $\phi$ in a neighborhood of the
constant stress-free value $\phi_0$:
\begin{equation}
    \phi(x)=\phi_0+\nu\phione(x),
    \label{phi-expansion}
\end{equation}
where $\nu>0$ is a `small' parameter to be identified from Problem \eqref{bvp}.
Correspondingly, also $c$ is expanded with respect to $\nu$ in a similar way:
\begin{equation}
    c(x)=\czero(x)+\nu c^{(1)}(x).
    \label{c-expansion}
\end{equation}

Inserting Eqs. \eqref{phi-expansion}, \eqref{c-expansion} into the second of
Eqs. \eqref{bvp}, and retaining only the terms up to the zeroth order in
$\nu$, yields the following problem for $\czero$:
\begin{equation}
    \left\{
    \begin{array}{rcll}
        -\partial_x^2\czero+\alpha w^2\phi_0\czero & = & 0 &
            \text{in\ } \domain \\
        \\[-0.3cm]
        \czero & = & 1 & x=0 \\
        \\[-0.3cm]
        \partial_x\czero & = & 0 & x=1,
    \end{array}
    \right.
    \label{bvp-c0}
\end{equation}
whose solution reads explicitly
\begin{equation}
    \czero_w(x)=\frac{\cosh{\left(w\sqrt{\alpha\phi_0}(1-x)\right)}}
        {\cosh\left(w\sqrt{\alpha\phi_0}\right)},
    \label{czero}
\end{equation}
where we have emphasized the dependence of $\czero$ on $w$ as a parameter. It is
an easy task to check that $\czero_w\in U$, indeed Problem \eqref{bvp-c0} falls
into Proposition \ref{wellpos-c} with the special choice $\varphi\equiv\phi_0\in
\V$.

In order to get an equation satisfied by $\phione$, and to detect at the same
time the parameter $\nu$, we proceed in a similar fashion, introducing
expressions \eqref{phi-expansion}, \eqref{c-expansion} into the
first of Eqs. \eqref{bvp}. Assuming formal differentiability of $g,\,\Gamma$ we
find, after some standard algebra,
$$ -\nu F'(\phi_0)\partial_x^2\phione+o(\nu)=
    w^2g(\phi_0)\Gamma(\czero_w)+O(\nu). $$
As $\Gamma$ is not identically zero, another term of zeroth order in $\nu$
should be found in this equation, besides the one at the right-hand side.
Estimating the order of magnitude of $\Gamma(\czero_w)$ via $\Gmax$ and looking
then at the left-hand side we discover $\nu=O(w^2\Gmax/F'(\phi_0))$. Moreover,
recalling the expression \eqref{Lf} for the Lipschitz constant $\Lf$ of $f$ on
$[F(\ep),\,1]$, we see that, at least for suitable functions $F$, we may
have $F'(\phi_0)=O(\Lf^{-1})$, hence
\begin{equation}
    \nu=O(w^2\Gmax\Lf)=O({(\beta_2w)}^2).
    \label{order_nu}
\end{equation}
As anticipated above, the parameter $\nu$ must be `small' in order for the
perturbative expansion to be meaningful. Equation \eqref{order_nu} shows that
this requirement corresponds in essence to the same condition prescribed by
Theorem \ref{solution} for the existence of solutions to Problem \eqref{bvp}.
Therefore, the existence theory and the perturbative approach are consistent
with one another, as they impose the same conditions on the parameters $\beta$,
$w$ of the model. In other words, for a certain value of the quantity $\beta
w$, they either apply or fail both.

Retaining only the terms up to the zeroth order in
$\nu$, the problem for the perturbation $\phione$ reads:
\begin{equation}
    \left\{
    \begin{array}{rcll}
        -\partial_x^2\phione & = & Cg(\phi_0)\Gamma(\czero_w) &
            \text{in\ } \domain \\
        \\[-0.3cm]
        \partial_x\phione & = & 0 & x=0 \\
        \\[-0.3cm]
        \phione & = & 0 & x=1,
    \end{array}
    \right.
    \label{bvp-phione}
\end{equation}
where $C=O(\Gmax^{-1})$ is an appropriate positive constant. In particular,
boundary conditions for $\phione$ are deduced from those prescribed on $\phi$
taking Eq. \eqref{phi-expansion} into account.

Unlike the previous case for $\czero_w$, Problem \eqref{bvp-phione} cannot
be solved in a closed form for a generic function $\Gamma$. However, classical
theory for linear elliptic equations guarantees existence and uniqueness of a
solution $\phione\in H^1_{0,1}(\domain)$, which is actually continuous due to
Sobolev embedding theorems for $\domain\subseteq\R$. Furthermore, since
$\czero_w\in U$ and $\Gamma$ is continuous on $[0,\,1]$, we even get
$\partial_x^2\phione\in\C$, hence $\phione$ is a classical solution to Problem
\eqref{bvp-phione}. This allows us to evaluate its derivative for
$x\in\bar\domain$ by integrating once the differential equation in
\eqref{bvp-phione}:
$$ -\partial_x\phione(x)+\partial_x\phione(0)=
    Cg(\phi_0)\intgrl_0^x\Gamma(\czero_w)\,d\xi, $$
whence, owing to the boundary conditions for $x=0$, $x=1$,
\begin{equation}
    \intgrl_0^1\Gamma(\czero_w)\,dx=0.
    \label{rootsGamma}
\end{equation}
Solving the (approximate) free boundary problem corresponds therefore to
finding the roots $w>0$ of Eq. \eqref{rootsGamma}.

After some preliminary considerations about the dependence of $\czero_w$
on $w$:
\begin{proposition} \label{c0-w}
Let $\czero_w\in\U$ be the solution to Problem \eqref{bvp-c0}. Then:
\begin{enumerate}
\item [(i)] The mapping $w\mapsto\czero_w(x)$ for $w\in\R_+$, $x\in\bar\domain$
is continuous and nonincreasing.

\item [(ii)] For $w\to+\infty$ the functions $\czero_w$ tend pointwise to the
function
$$ \czero_\infty(x)=
    \begin{cases}
        1 & \text{if\ } x=0 \\
        \\[-0.3cm]
        0 & \text{if\ } x\ne 0.
    \end{cases}
$$

\item [(iii)] For all $c_0\in(0,\,1)$ there exists $w_\ast>0$ such that if
$w\geq w_\ast$ the equation
$$ \czero_w(x)=c_0 $$
in the unknown $x$ has exactly one solution $\bar x_w\in\bar\domain$.

\item [(iv)] For any fixed $c_0\in(0,\,1)$ one has $\bar x_w\to 0$ when
$w\to+\infty$.
\end{enumerate}
\end{proposition}
\noindent we have all we need to solve our free boundary problem. The following
result should be compared to the analogous one proved by Bueno \emph{et al.}
\cite{MR2150346}.
\begin{theorem} \label{w0-solfb}
Let Assumptions \ref{assGamma}, \ref{assGamma-2} hold. Then there exists a
positive solution $w_0$ to Eq. \eqref{rootsGamma}.
\end{theorem}

Notice that Theorem \ref{w0-solfb} does not guarantee that the solution $w_0$
actually satisfies $\beta w_0<1$. In practice, this condition has to be checked
\emph{a posteriori}, after solving Eq. \eqref{rootsGamma} explicitly. Concerning
this, we observe that, once the function $\Gamma$ has been specified, Eq.
\eqref{rootsGamma} often specializes in an algebraic equation in the unknown
$w$, that can be solved with arbitrary precision, possibly with the aid of
suitable numerical techniques. After getting the solution $w_0$ claimed by
Theorem \ref{w0-solfb}, two situations may arise: Condition $\beta w_0<1$ is
either satisfied, hence, in view of the reasonings above, one can take $w_0$ as
a good approximation of the solution to the free boundary problem, or it is
violated, in which case one should reject $w_0$ and conclude that, for the
specific model parameters at hand, no information on the existence of a
solution satisfying the free boundary condition can be obtained from the
present theory.

We show an example of application of this procedure in the next Sect.
\ref{notex}.

\subsection{A noticeable example}
\label{notex}
In order to illustrate how the previous theory applies to a specific model, we
consider for $\mu\geq 1$  the following class of intercellular stress functions
$\Sigma$ (cf. Eq. \eqref{sigma-pormed}):
\begin{equation}
    \Sigma(\phi)=
        \begin{cases}
            \dfrac{1}{\phi}\log{\dfrac{\phi}{\phi_0}} & \text{if\ } \mu=1 \\
            \\[-0.3cm]
            \dfrac{\mu}{\mu-1}\cdot\dfrac{\phi^{\mu-1}-\phi_0^{\mu-1}}{\phi}
                & \text{if\ } \mu>1, \\
        \end{cases}
        \label{sigma-pormed-2}
\end{equation}
whose behavior for several $\mu$ is illustrated in Fig. \ref{fig-sigma}. Notice
that for $\mu>2$ the function $\Sigma$ grows unboundedly with $\phi$, and is
thus compatible with the expected behavior of an intercellular pressure-like
stress, while for $1\leq\mu\leq 2$ it is bounded from above, and if $\mu\ne 2$
it even tends to zero for $\phi\to +\infty$. From the modeling point of view,
the meaningful cases correspond to $\mu>2$. However, the theory developed in
the previous sections also covers the range of values $1\leq\mu\leq 2$, which
may be of some theoretical interest due to the particular form of the equations
they originate (for instance, for $\mu=1$ the differential part of the equation
for $\phi$ is linear).

According to Eq. \eqref{F-Sigma}, the corresponding generalized stress function
$F$ is
$$ F(\phi)=\phi^{\mu}, $$
which gives rise to the stationary porous medium equation for the cell volume
ratio $\phi$. It is straightforward to check that such an $F$ complies with
Assumption \ref{assF} for all $\mu\geq 1$. Moreover, we choose for $g$
and $\Gamma$ the expressions given by Eqs. \eqref{g}, \eqref{Gamma}:
$$ g(\phi)=\phi(1-\phi), \qquad \Gamma(c)=\gamma(c-c_0), $$
where $0<c_0<1$ and $\gamma>0$ are constant, whence also Assumptions
\ref{assg}, \ref{assGamma}, \ref{assGamma-2} are satisfied.

\begin{figure}[t]
    \begin{center}
        \includegraphics[height=7cm]{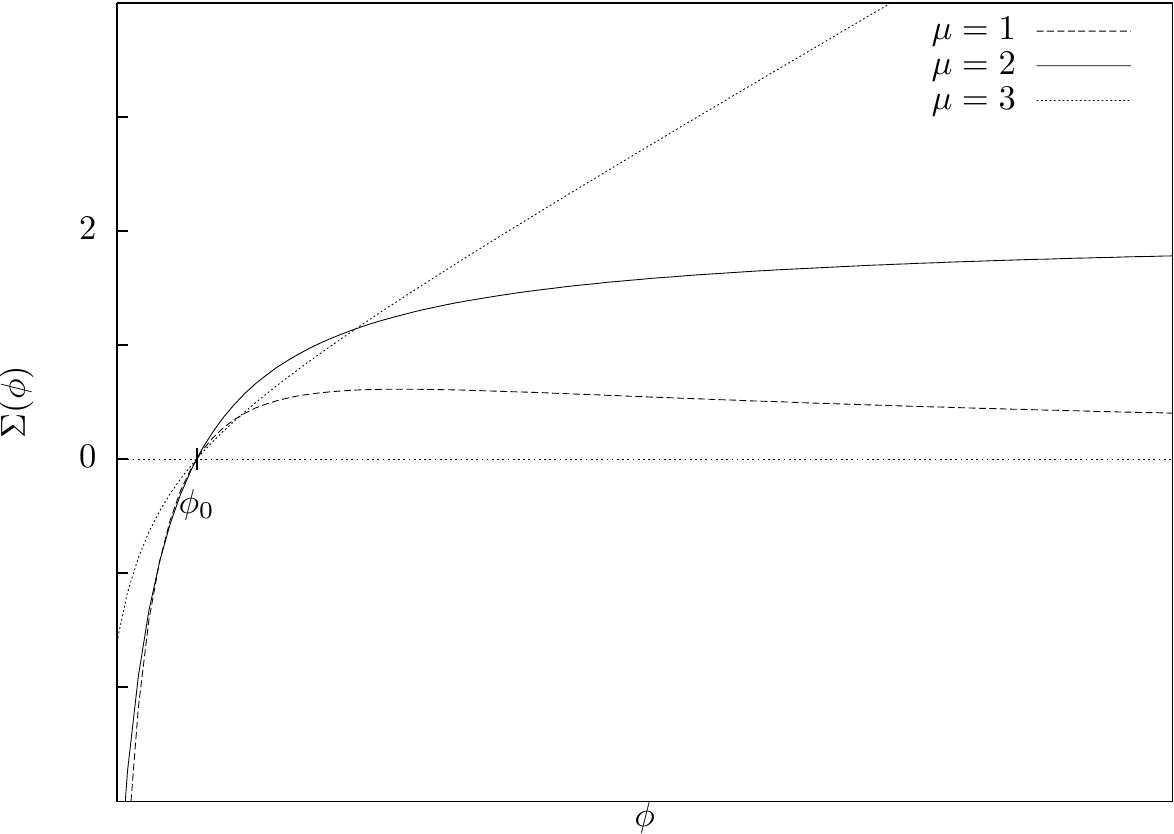}
    \end{center}
    \caption{Cell stress function $\Sigma(\phi)$ of Eq. \eqref{sigma-pormed-2}
        for $\mu=1,\,2,\,3$.}
    \label{fig-sigma}
\end{figure}

The inverse function $f$ of $F$ on $[0,\,1]$ is
$$ f(\phi)=\sqrt[\mu]{\phi}. $$
Notice that $f$ is actually Lipschitz on every compact subset of $(0,\,1]$ for
all $\mu\geq 1$, but for $\mu>1$ it is not Lipschitz on the whole interval
$[0,\,1]$ due to the vertical tangent at $\phi=0$. However, we know that the
interesting parameter is its Lipschitz constant $\Lf$ on $[F(\ep),\,1]$,
which exists for any $\mu\geq 1$ and equals
$$ \Lf=\max_{\xi\in[F(\ep),\,1]}f'(\xi)
    =\max_{\xi\in[\ep^\mu,\,1]}\frac{1}{\mu\sqrt[\mu]{\xi^{\mu-1}}}
    =\frac{1}{\mu\ep^{\mu-1}}. $$
Moreover,
\begin{gather*}
    \gmax=\max_{\xi\in[0,\,1]}\vert g(\xi)\vert=\frac{1}{4}, \qquad
        \Lg=\max_{\xi\in[0,\,1]}\vert g'(\xi)\vert=1, \\
        \Gmax=\max_{\xi\in[0,\,1]}\vert\Gamma(\xi)\vert=
        \gamma\max{(c_0,\,1-c_0)},
\end{gather*}
whence we compute
$$ \beta_1(\ep)=\sqrt{\frac{\gamma\max{(c_0,\,1-c_0)}}
    {2\pi\left(\phi_0^\mu-\ep^\mu\right)}}, \qquad
    \beta_2(\ep)=\frac{2}{\pi}\sqrt{\frac{\gamma\max{(c_0,\,1-c_0)}}
    {\mu\ep^{\mu-1}}}. $$
For all $w>0$ such that $\beta w<1$, that is
\begin{equation}
    w<\min{\left(\frac{1}{\beta_1},\,\frac{1}{\beta_2}\right)},
    \label{bound-w}
\end{equation}
Theorem \ref{solution} guarantees the existence of a solution
$(\phi,\,c)\in\V\times\U$ to the problem
\begin{equation}
    \left\{
    \begin{array}{ll}
        -\partial_x^2\phi^\mu =
            \gamma w^2\phi(1-\phi)(c-c_0) & \text{in\ } \domain \\
        \\[-0.3cm]
        -\partial_x^2c+\alpha w^2\phi c=0 & \text{in\ } \domain \\
        \\[-0.3cm]
        \partial_x\phi=0,\ c=1 & x=0 \\
        \\[-0.3cm]
        \phi=\phi_0,\ \partial_xc=0 & x=1.
    \end{array}
    \right.
    \label{bvp-porousmedium}
\end{equation}

\begin{figure}[t]
    \begin{center}
        \includegraphics[width=5.9cm]{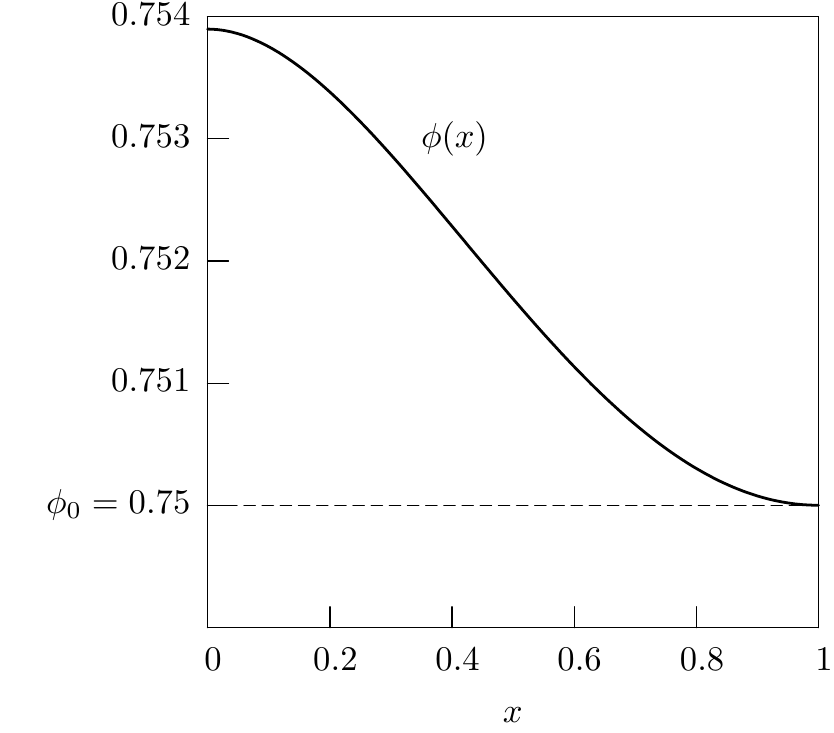}
        \qquad
        \includegraphics[width=5.9cm]{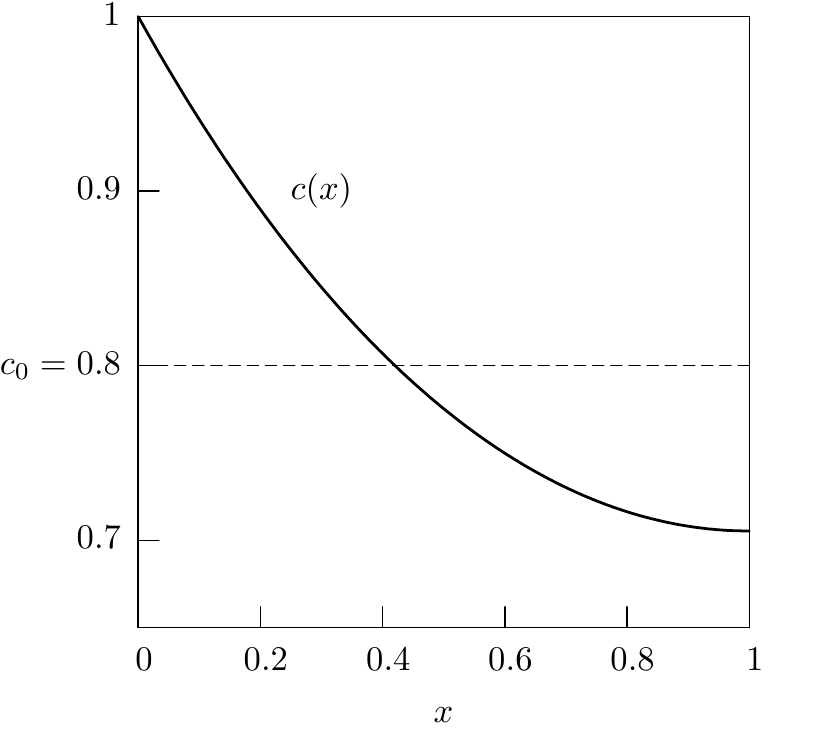}
    \end{center}
    \caption{Numerical solution (FEM) to Problem \eqref{bvp-porousmedium}
    fulfilling the free boundary condition \eqref{bc-freebound}. Cell density
    $\phi$ on the left, nutrient concentration $c$ on the right.}
    \label{fig-numsim}
\end{figure}

Let us now look for a solution to the free boundary problem \eqref{bc-freebound}
via the small parameter approximation discussed in Subsect. \ref{solfreebound}.
The specific form of the function $\Gamma$ we are using allows for an easy
estimate of the solutions to Eq. \eqref{rootsGamma}. Inserting the expression of
$\czero_{w}$ given by Eq. \eqref{czero} into the latter gives
\begin{equation}
    \tanh{\left(w\sqrt{\alpha\phi_0}\right)}=w\sqrt{\alpha\phi_0}c_0,
    \label{eq-w}
\end{equation}
thus, invoking McLaurin expansion of the hyperbolic tangent and considering
that $\tanh(x)\leq 1$ for all $x\in\R$,
\begin{equation}
    \sqrt{\frac{3(1-c_0)}{\alpha\phi_0}}\leq w_0
        \leq\frac{1}{c_0\sqrt{\alpha\phi_0}}.
    \label{estimate-w}
\end{equation}

For the sake of definiteness, let us fix the same parameters already used in Eq.
\eqref{nondim_param-numbers}, then let us determine the optimal $\ep$ by
solving the min-max problem
$$ \min_{\ep\in(0,\,\phi_0)}\max_{i=1,\,2}\beta_i(\ep). $$
This gives $\ep\simeq 0.5$ and correspondingly $\beta_1\simeq\beta_2\simeq
0.55$, whence it can be easily computed that the estimate \eqref{estimate-w}
agrees with the previous bound \eqref{bound-w}. Solving Eq. \eqref{eq-w} via an
iterative numerical procedure yields $w_0\simeq 1.45$. Since $\beta w_0\simeq
0.80$, the small parameter assumption and the related perturbative approach are
consistent with the problem at hand, as it is definitely confirmed by the value
of $w$ returned by the Finite Element solution (see Fig. \ref{fig-density_2D},
\ref{fig-numsim}) of the exact problem, i.e., $w\simeq 1.44$. The behaviors of
the relative errors on $\phi$ and $c$ (see Figure \ref{fig-relerr}):
\begin{equation*}
    E[\phi](x):=1-\frac{\phi_0+\nu\phione(x)}{\phi(x)}, \qquad
    E[c](x):=1-\frac{\czero_w(x)}{c(x)},
\end{equation*}
where $\phione$ is analytically computed for this particular case from Eq.
\eqref{bvp-phione} using $C=\Gmax^{-1}$, and where we have set $\nu={(\beta_2
w)}^2$, further prove the good quality of the small parameter approximation with
respect to the exact solution.

\begin{figure}[t]
    \begin{center}
        \includegraphics[width=5.9cm]{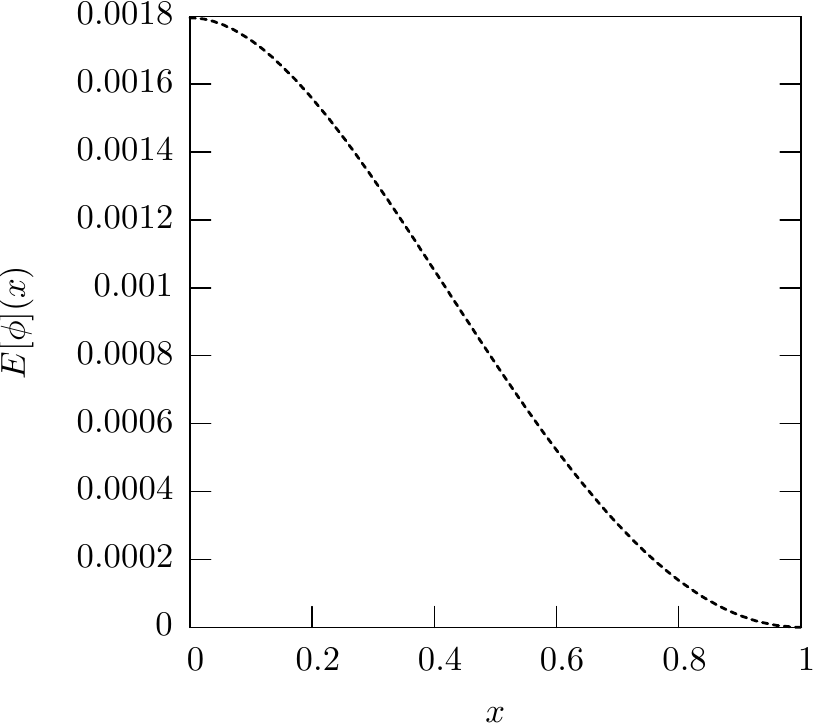}
        \qquad
        \includegraphics[width=5.9cm]{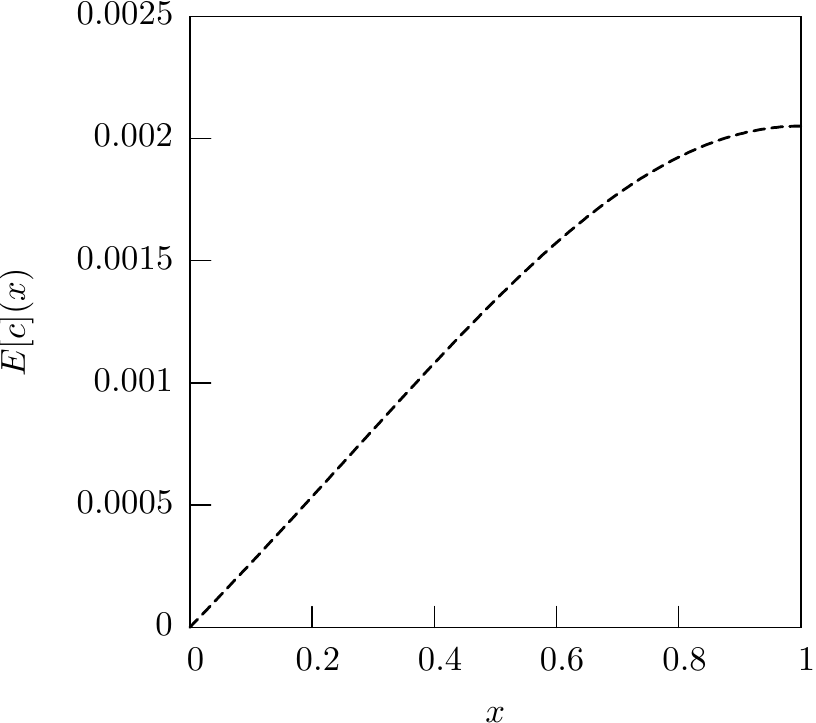}
    \end{center}
    \caption{Relative errors on $\phi$ (left) and $c$ (right) as they result
    from comparison between the approximate solution, analytically computed via
    Eqs. \eqref{czero}, \eqref{bvp-phione}, and the numerical solution (FEM) to
    the exact equations (Problem \eqref{bvp-porousmedium}). Notice that in
    both cases one has $E=O(10^{-3})$.}
    \label{fig-relerr}
\end{figure}

\section{Technical proofs}
\label{sect-proofs}
\setcounter{theorem}{-1}

Here we collect, without specific comments, the proofs of all results stated in
the previous Sect. \ref{sect-statprob}. Before going into technical details,
however, it is convenient to fix some basic facts and notations that we will
extensively use in the sequel.

We denote by $H^1_{0,0}(\domain)$ ($H^1_{0,1}(\domain)$, respectively) the
closed linear subspace of the Sobolev space $H^1(\domain)=W^{1,2}(\domain)$
consisting of all functions $u\in H^1(\domain)$ whose trace vanishes at $x=0$
($x=1$, respectively). We recall that Poincar\'e inequality holds true for
functions $u\in H^1_{0,0}(\domain)$, or $u\in H^1_{0,1}(\domain)$:
$$ \|u\|_{\L2}\leq C_P\|\partial_xu\|_{\L2}, $$
where the Poincar\'e constant $C_P$ equals $\frac{2}{\pi}$. As a consequence,
the quantity
$$ \|u\|_{\H}:=\|\partial_xu\|_{\L2} $$
defines a norm in both $H^1_{0,0}(\domain)$ and $H^1_{0,1}(\domain)$
equivalent to the usual $H^1$-norm. Notice that in referring to such a norm
we use generically the short subscript $\H$ for either of the spaces, the
meaning being recoverable each time from the context.

When dealing with continuous functions $u\in\C$, we indicate by $\|u\|_\infty$
their $\infty$-norm over $\bar\domain$:
$$ \|u\|_\infty=\max_{x\in\bar\domain}\vert u(x)\vert. $$

Finally, we use the symbols $u^+,\,u^-$ for the positive and negative
part of a function $u$:
$$ u^+=\max{(u,\,0)}, \qquad u^-=\max{(-u,\,0)}. $$
In view a theorem due to G. Stampacchia, it is known that $u\in H^1(\domain)$
implies $u^+,\,u^-\in H^1(\domain)$ with moreover
$$  \partial_xu^+=
    \begin{cases}
        \partial_xu & \text{if\ } u\geq 0 \\
        0 & \text{otherwise},
    \end{cases}
    \qquad
    \partial_xu^-=
    \begin{cases}
        -\partial_xu & \text{if\ } u\leq 0 \\
        0 & \text{otherwise}.
    \end{cases}
$$

Let us start by a mainly technical Lemma, which will be sometimes referenced in
the forthcoming proofs.
\begin{lemma}
\label{lemma-uhk}
Let $h,\,k\in\C$ with moreover $h(x)\geq 0$ for all $x\in\bar\domain$. If
$u\in\C$ solves the equation
$$ -\partial_x^2u+h(x)u=k(x) \quad \text{in\ } \domain $$
with either of the following sets of boundary conditions:
\begin{enumerate}
\item [(i)] $u=0$ for $x=0$; $\partial_xu=0$ for $x=1$, or;
\item [(ii)] $\partial_xu=0$ for $x=0$; $u=0$ for $x=1$,
\end{enumerate}
then
$$ \vert u(x)\vert\leq\|k\|_\infty, \quad \forall\,x\in\bar\domain. $$
\end{lemma}
\begin{proof}
Using standard techniques for linear elliptic problems, it can be shown that the
proposed equation admits a unique solution $u\in H^1_{0,0}(\domain)$ in case
(i), and a unique solution $u\in H^1_{0,1}(\domain)$ in case (ii). In both
cases, $u\in\C$ due to Sobolev embedding theorems for $I\subseteq\R$. Moreover,
Lax-Milgram Theorem entails the following \emph{a priori} estimate:
$$ \|u\|_{\H}\leq\|k\|_{\L2}, $$
which, considering that
$$ \|k\|_{\L2}^2=\intgrl_0^1{\vert k(x)\vert}^2\,dx\leq
    \intgrl_0^1{\left(\max_{x\in\bar\domain}\vert k(x)\vert\right)}^2\,dx=
    \|k\|_\infty^2 $$
and, from Morrey inequality, that $\|u\|_\infty\leq\|u\|_{\H}$, yields
$$ \|u\|_\infty\leq\|k\|_\infty $$
and thus the thesis.
\end{proof}

\begin{proposition}
To every $\varphi\in\V$ there exists a unique solution $c\in\U$ of Problem
\eqref{bvp-c}, which is further twice continuously differentiable and
monotonically decreasing on $\bar\domain$, with a maximum value equal to $1$
attained for $x=0$.
\end{proposition}
\begin{proof}
\begin{enumerate}
\item By means of the substitution $\bar{c}=c-1$, Problem \eqref{bvp-c} is
converted into the following linear elliptic problem with homogeneous boundary
conditions:
\begin{equation*}
    \left\{
    \begin{array}{rcll}
        -\partial_x^2\bar{c}+\alpha w^2\varphi\bar{c} & = &
            -\alpha w^2\varphi & \text{in\ }\domain \\
        \\[-0.3cm]
        \bar{c} & = & 0 & x=0 \\
        \\[-0.3cm]
        \partial_x\bar{c} & = & 0 & x=1.
    \end{array}
    \right.
\end{equation*}
By putting it in weak form as
\begin{equation}
    \begin{cases}
        \text{find\ } \bar{c}\in H^1_{0,0}(I)\ \text{such that} \\
        \\[-0.3cm]
        a(\bar{c},\,v)=b(v),\quad \forall\,v\in H^1_{0,0}(I),
    \end{cases}
    \label{bvp-weak-sigma}
\end{equation}
where the bilinear form $a:H^1_{0,0}(\domain)\times H^1_{0,0}(\domain)\to\R$ and
the linear form $b:H^1_{0,0}(\domain)\to\R$ are defined by
$$ a(u,\,v)=\intgrl_0^1\partial_xu\,\partial_xv\,dx+\alpha w^2\intgrl_0^1
    \varphi uv\,dx, \qquad
    b(v)=-\alpha w^2\intgrl_0^1\varphi v\,dx $$
respectively, classical theory for linear elliptic PDEs can be applied up to
checking that $a$ is continuous and coercive on
$H^1_{0,0}(\domain)\times H^1_{0,0}(\domain)$ and that $b$ is continuous on
$H^1_{0,0}(\domain)$. Owing to Lax-Milgram Theorem, we get existence and
uniqueness of a solution $\bar{c}\in H^1_{0,0}(\domain)$ to Problem
\eqref{bvp-weak-sigma}, and by consequence of a solution $c\in H^1(\domain)$ to
Problem \eqref{bvp-c}. Moreover, for $\domain\subseteq\R$ Sobolev embedding
theorems imply $H^1(\domain)\subset\C$, thus $c$ is in fact continuous
on $\bar{\domain}$.

\item We check now that $c\geq 0$ on $\bar\domain$. For this, observe first
that
$c^-\in H^1_{0,0}(\domain)$, then multiply Eq. \eqref{bvp-c} by $c^-$ and
integrate by parts over $\domain$:
$$ -\|c^-\|_{\H}^2-\alpha w^2\intgrl_0^1\varphi{(c^-)}^2\,dx=0 $$
whence
$$ \|c^-\|_{\H}^2=-\alpha w^2\intgrl_0^1\varphi{(c^-)}^2\,dx\leq 0. $$
We deduce $\|c^-\|_{\H}=0$, which yields $c^-=0$ almost everywhere in $\domain$
and, by continuity, $c(x)\geq 0$ for all $x\in\bar\domain$.

\item Finally, we show that $c\leq 1$ on $\bar\domain$. Since $c\in
C(\bar\domain)$, the second order derivative $\partial_x^2 c$ exists as a
distribution. But from Eq. \eqref{bvp-c} we discover that it actually coincides
with a continuous function:
$$ \partial_x^2c=\alpha w^2\varphi c\in \C, $$
so that $c$ turns out to be a classical solution of Problem \eqref{bvp-c}. This
allows us to compute, for every $x\in\bar\domain$,
$$ (\partial_xc)(1)-(\partial_xc)(x)=\alpha w^2\intgrl_x^1\varphi c\,d\xi, $$
whence, using the boundary condition at $x=1$,
$$ (\partial_xc)(x)=-\alpha w^2\intgrl_x^1\varphi c\,d\xi\leq 0, \qquad
    \forall\,x\in\bar\domain. $$
Therefore we have that $c$ is nonincreasing on $\bar\domain$, and consequently
we deduce
$$ \max_{x\in\bar\domain}c(x)=c(0)=1, $$
which completes the proof. \qedhere
\end{enumerate}
\end{proof}

\begin{proposition}
Set
\begin{equation*}
    \beta_1=\beta_1(\ep):=
        \sqrt{\frac{C_P\gmax\Gmax}{F(\phi_0)-F(\ep)}}, \qquad
    \beta_2=\beta_2(\ep):=C_P\sqrt{\Lg\Gmax\Lf}
\end{equation*}
and define
\begin{equation*}
    \beta=\beta(\ep):=\max{\left(\beta_1(\ep),\,
        \beta_2(\ep)\right)}.
\end{equation*}
If $\beta w<1$, then for every $\sigma\in\U$ Problem \eqref{bvp-phi} admits a
unique solution $\phi\in\V$.
\end{proposition}
\begin{proof}
\begin{enumerate}
\item First we rewrite Problem \eqref{bvp-phi} making the substitution
\begin{equation}
    u=F(\phi)-F(\phi_0),
    \label{subst}
\end{equation}
which in turn implies $\phi=f(u+u_0)$ for $u_0:=F(\phi_0)\in (0,\,1)$:
\begin{equation}
    \left\{
    \begin{array}{rcll}
        -\partial_x^2u & = & w^2g(f(u+u_0))\Gamma(\sigma)
        & \text{in\ } \domain \\
        \\[-0.3cm]
        \partial_xu & = & 0 & x=0 \\
        \\[-0.3cm]
        u & = & 0 & x=1.
    \end{array}
    \right.
    \label{bvp-phi-2}
\end{equation}
Notice that for $\phi\in\V$ the boundary conditions at $x=0$ of Problems
\eqref{bvp-phi} and \eqref{bvp-phi-2} agree, since
$\partial_xu=F'(\phi)\partial_x\phi$ with $F'(\phi)>0$.

\item Now we prove that for all $w>0$ Problem \eqref{bvp-phi} admits continuous
solutions $\phi$ ranging in $[0,\,1]$. In view of Eq. \eqref{subst} and of the
continuity and monotonicity properties of $F$, this amounts to showing that
Problem \eqref{bvp-phi-2} possesses solutions $u\in\C$ bounded between $-u_0$
and $1-u_0$.

\begin{enumerate}
\item Let us introduce the function
$$ \tilde g(u)={[g(f(u+u_0))]}^+ $$
and consider then the auxiliary problem
\begin{equation}
    \left\{
    \begin{array}{rcll}
        -\partial_x^2u & = & w^2\tilde g(u)\Gamma(\sigma)
        & \text{in\ } \domain \\
        \\[-0.3cm]
        \partial_xu & = & 0 & x=0 \\
        \\[-0.3cm]
        u & = & 0 & x=1.
    \end{array}
    \right.
    \label{bvp-phi-aux}
\end{equation}
The function $\tilde g(u)$ being continuous, if $P(u)$ denotes any of its
antiderivatives on $\R$ then Problem \eqref{bvp-phi-aux} can be viewed as the
Euler-Lagrange equation associated to the energy functional
$E:H^1_{0,1}(\domain)\to\R$,
$$ E(u)=\frac{1}{2}\intgrl_0^1{(\partial_xu)}^2\,dx-
    w^2\intgrl_0^1P(u)\Gamma(\sigma)\,dx. $$
Due to the assumptions on the sign of $g$ and the monotonicity of $f$, it
results $\tilde g(u)=0$ for $u\leq -u_0$ or $u\geq 1-u_0$, hence $P(u)$ is
bounded: There exists a constant $C>0$ such that $\vert P(u)\vert\leq C$ for all
$u\in\R$. This entails first that $E$ is finite on $H^1_{0,1}(\domain)$:
$$ E(u)\leq\frac{1}{2}\|u\|_{\H}^2+w^2C\Gmax<+\infty, $$
and moreover that it is coercive:
\begin{align*}
    E(u) &\geq \frac{1}{2}\|u\|_{\H}^2-w^2\intgrl_0^1\vert P(u)\vert\cdot
        \vert\Gamma(\sigma)\vert\,dx \\
    &\geq \frac{1}{2}\|u\|_{\H}^2-w^2C\Gmax \\
    &\geq \|u\|_{\H}-C',
\end{align*}
where $C'>0$ is a new suitable constant. Thus any minimizing sequence
is bounded in $H^1_{0,1}(\domain)$, and has therefore weakly convergent
subsequences. If we can show that $E$ is sequentially weakly lower
semicontinuous on $H^1_{0,1}(\domain)$, we will deduce the existence of a
minimizer $u\in H^1_{0,1}(\domain)$, that is a solution to the auxiliary problem
\eqref{bvp-phi-aux}. For this, notice that $E$ belongs to the class of
functionals of the form
$$ E(u)=\intgrl_0^1 L(\partial_xu,\,u,\,x)\,dx, $$
where $L:\R^2\times\domain\to\R$ is the function given by
$$ L(p,\,z,\,x)=\frac{p^2}{2}-w^2P(z)\Gamma(\sigma(x)). $$
Clearly, $L$ is convex in $p$ for each $z\in\R$, $x\in\domain$, and moreover
it is bounded from below as it results $L(p,\,z,\,x)\geq -w^2C\Gmax$ for all
$(p,\,z,\,x)\in\R^2\times\domain$. These two features are sufficient in order
for $E$ to be sequentially weakly lower semicontinuous on $H^1(\domain)$ (see
Evans \cite{MR1625845}, Chapter 8, p. 446 for further details), hence also on
$H^1_{0,1}(\domain)$ as desired.

We conclude that the auxiliary problem \eqref{bvp-phi-aux} has solutions $u\in
H^1_{0,1}(\domain)$ which, as a consequence of Sobolev embedding theorems for
$I\subseteq\R$, are actually continuous functions $u\in\C$.

\item We claim that any solution $u\in H^1_{0,1}(\domain)$ to Problem
\eqref{bvp-phi-aux} satisfies the inequalities $-u_0\leq u(x)\leq 1-u_0$ for all
$x\in\bar\domain$. To see this, let us multiply Eq. \eqref{bvp-phi-aux} by
${(u+u_0)}^-\in H^1_{0,1}(\domain)$ and integrate by parts over $\domain$:
$$ \intgrl_0^1\partial_xu\,\partial_x{(u+u_0)}^{-}\,dx=
    w^2\intgrl_0^1\tilde g(u){(u+u_0)}^{-}\Gamma(\sigma)\,dx. $$
Since $\tilde g(u)=0$ for $u\leq -u_0$ while ${(u+u_0)}^-=0$ for $u\geq -u_0$,
this relation implies
$$ \|{(u+u_0)}^-\|_{\H}^2=0, $$
hence, by continuity, $u(x)\geq -u_0$ for all $x\in\bar\domain$.

Analogously, multiplying Eq. \eqref{bvp-phi-aux} by ${(1-u_0-u)}^-\in
H^1_{0,1}(\domain)$ and integrating by parts over $\domain$ gives
$$ \intgrl_0^1\partial_xu\,\partial_x{(1-u_0-u)}^-\,dx=
    w^2\intgrl_0^1\tilde g(u){(1-u_0-u)}^-\Gamma(\sigma)\,dx. $$
But $\tilde g(u)=0$ for $u\geq 1-u_0$ while ${(1-u_0-u)}^-=0$ for $u\leq 1-u_0$,
which entails
$$ \|{(1-u_0-u)}^-\|_{\H}^2=0 $$
and finally $u(x)\leq 1-u_0$ for all $x\in\bar\domain$.

\item We observe now that for $u\in[-u_0,\,1-u_0]$ one has
$\tilde g(u)=g(f(u+u_0))$, thus any solution to the auxiliary problem
is also a solution to Problem \eqref{bvp-phi-2}. Therefore, the latter admits
solutions $u\in\C$ bounded between $-u_0$ and $1-u_0$, which, via Eq.
\eqref{subst}, originate solutions $\phi\in\C$ to Problem \eqref{bvp-phi} such
that $0\leq\phi\leq 1$ on $\bar\domain$.
\end{enumerate}

\item Finally, we show that it is possible to get solutions $\phi\in\V$ to
Problem \eqref{bvp-phi} provided $w$ is sufficiently small. For this, let
$\phi\in\C$, $0\leq\phi\leq 1$, be a solution to Problem \eqref{bvp-phi}. Since
$F(\phi_0)$ is constant, we can rewrite the equation for $\phi$ in the
equivalent form
$$ -\partial_x^2\left(F(\phi)-F(\phi_0)\right)=w^2g(\phi)\Gamma(\sigma). $$
Observe now that $F(\phi)-F(\phi_0)\in H^1_{0,1}(\domain)$, with moreover
$$ \partial_x\left(F(\phi)-F(\phi_0)\right)=F'(\phi)\partial_x\phi=0 $$
at $x=0$. Thus, multiplying both sides of the above equation by
$F(\phi)-F(\phi_0)$ and integrating by parts over $\domain$, we discover
$$ \|F(\phi)-F(\phi_0)\|_{\H}\leq w^2C_P\gmax\Gmax, $$
where we have used the fact that, for $\phi\in[0,\,1]$, the function $g$ is
bounded by $\gmax$. In view of Morrey inequality we further deduce
$$ \|F(\phi)-F(\phi_0)\|_\infty\leq w^2C_P\gmax\Gmax, $$
whence
$$ F(\phi_0)-w^2C_P\gmax\Gmax\leq F(\phi)\leq F(\phi_0)+w^2C_P\gmax\Gmax $$
on $\bar\domain$. Condition $\phi\geq\ep$ is equivalent to $F(\phi)\geq
F(\ep)$ due to the monotonicity of $F$. Imposing
$$ F(\ep)\leq F(\phi_0)-w^2C_P\gmax\Gmax, $$
we obtain $\phi\in\V$ provided $\beta_1 w\leq 1$, where
$$ \beta_1=\beta_1(\ep)=
    \sqrt{\frac{C_P\gmax\Gmax}{F(\phi_0)-F(\ep)}}. $$
Notice that $\beta_1$ is well defined, since $\ep<\phi_0$ implies
$F(\ep)<F(\phi_0)$.

\item Regarding uniqueness, suppose, under the hypothesis $\beta_1w\leq 1$, that
$\phi_1,\,\phi_2\in\V$ are two solutions to Problem \eqref{bvp-phi}. Subtracting
the respective equations we find
$$ -\partial_x^2(F(\phi_2)-F(\phi_1))=
    w^2[g(\phi_2)-g(\phi_1)]\Gamma(\sigma), $$
where $F(\phi_2)-F(\phi_1)\in H^1_{0,1}(\domain)$ is such that
$$ \partial_x(F(\phi_2)-F(\phi_1))=F'(\phi_2)\partial_x\phi_2-
F'(\phi_1)\partial_x\phi_1=0 $$ for $x=0$. Multiplying both sides by
$F(\phi_2)-F(\phi_1)$ and integrating by parts over $\domain$ we get
\begin{align*}
    \|F(\phi_2)-F(\phi_1)\|_{\H}^2 &=w^2\intgrl_0^1
    [g(\phi_2)-g(\phi_1)](F(\phi_2)-F(\phi_1))\Gamma(\sigma)\,dx \\
    &\leq w^2\Lg\Gmax \intgrl_0^1\vert\phi_2-\phi_1\vert\cdot
        \vert F(\phi_2)-F(\phi_1)\vert\, dx \\
    &= w^2\Lg\Gmax\intgrl_0^1\vert fF(\phi_2)-fF(\phi_1)\vert\cdot
        \vert F(\phi_2)-F(\phi_1)\vert\, dx \\
    &\leq w^2\Lg\Gmax\Lf C_P^2\|F(\phi_2)-F(\phi_1)\|_{\H}^2,
\end{align*}
where we have used the Lipschitz continuity of $g$ on $[0,\,1]$. Setting
$$ \beta_2=\beta_2(\ep)=C_P\sqrt{\Lg\Gmax\Lf}, $$
we deduce therefore
$$ \left(1-\beta_2^2w^2\right)\|F(\phi_2)-F(\phi_1)\|_{\H}^2\leq 0, $$
thus if $w$ satisfies the further constraint $\beta_2w<1$ we conclude
$$ \|F(\phi_2)-F(\phi_1)\|_{\H}=0, $$
that is $F(\phi_2)=F(\phi_1)$ for all $x\in\bar\domain$, and the uniqueness
follows from the bijectivity of $F$.

\item In view of the previous results and given the definition of $\beta$, we
conclude that if $\beta w<1$ then there exists a unique solution $\phi\in\V$ to
Problem \eqref{bvp-phi} as desired. \qedhere
\end{enumerate}
\end{proof}

\begin{proposition}
The operators $A_1:\V\to\U$, $A_2:\U\to\V$ are Lipschitz continuous, and so is
$A:\V\to\V$.
\end{proposition}
\begin{proof}
\begin{enumerate}
\item Let us begin by considering $A_1$. Given $\varphi_1,\,\varphi_2\in\V$, let
$c_1=A_1(\varphi_1)$, $c_2=A_1(\varphi_2)$ be the corresponding solutions in
$\U$ to Problem \eqref{bvp-c} and set $u=c_2-c_1\in\C$. Then $u$ solves the
problem
\begin{equation*}
    \left\{
    \begin{array}{rcll}
        -\partial_x^2u+\alpha w^2\varphi_2u & = &
            -\alpha w^2c_1(\varphi_2-\varphi_1) & \text{in\ } \domain \\
        \\[-0.3cm]
        u & = & 0 & x=0 \\
        \\[-0.3cm]
        \partial_xu & = & 0 & x=1.
    \end{array}
    \right.
\end{equation*}
Setting
$$ h(x)=\alpha w^2\varphi_2(x), \qquad
    g(x)=-\alpha w^2c_1(x)(\varphi_2(x)-\varphi_1(x)), $$
we have $h,\,g\in\C$, $h\geq 0$ on $\bar\domain$, and moreover
$$ \|g\|_\infty\leq\alpha w^2\|\varphi_2-\varphi_1\|_\infty $$
because $c_1\in\U$ satisfies $\vert c_1\vert\leq 1$ on $\bar\domain$. >From Lemma
\ref{lemma-uhk} we deduce then
$$ \|A_1(\varphi_2)-A_1(\varphi_1)\|_\infty=\|c_2-c_1\|_\infty=\|u\|_\infty
    \leq\|g\|_\infty\leq\alpha w^2\|\varphi_2-\varphi_1\|_\infty, $$
whence the thesis follows.

\item Concerning $A_2$, let $\sigma_1,\,\sigma_2\in\U$ and let
$\phi_1=A_2(\sigma_1)$, $\phi_2=A_2(\sigma_2)$ be the solutions in $\V$ to
Problem \eqref{bvp-phi}, respectively. Subtracting the corresponding equations
we find
\begin{align*}
    -\partial_x^2(F(\phi_2)-F(\phi_1)) &= w^2\left[g(\phi_2)
        \Gamma(\sigma_2)-g(\phi_1)\Gamma(\sigma_1)\right] \\
    &= w^2\left[g(\phi_2)-g(\phi_1)\right]\Gamma(\sigma_2)+
        w^2g(\phi_1)[\Gamma(\sigma_2)-\Gamma(\sigma_1)].
\end{align*}
After observing that $F(\phi_2)-F(\phi_1)\in H^1_{0,1}(\domain)$,
$\partial_x(F(\phi_2)-F(\phi_1))=0$ for $x=0$, we multiply both sides by
$F(\phi_2)-F(\phi_1)$ and integrate by parts over $\domain$, like in the proof
of Proposition \ref{wellpos-phi}, to obtain
\begin{align*}
\|F(\phi_2)-F(\phi_1)\|_{\H}^2 &\leq
    w^2\Gmax\intgrl_0^1\vert g(\phi_2)-g(\phi_1)\vert\cdot
        \vert F(\phi_2)-F(\phi_1)\vert\,dx \\
    &\phantom{\leq}+w^2\intgrl_0^1\vert g(\phi_1)\vert\cdot
        \vert\Gamma(\sigma_2)-\Gamma(\sigma_1)\vert\cdot\vert
F(\phi_2)-F(\phi_1)\vert\,
        dx \\
    &\leq \beta_2^2w^2\|F(\phi_2)-F(\phi_1)\|_{\H}^2 \\
    &\phantom{\leq} +w^2\gmax L_\Gamma C_P\|\sigma_2-\sigma_1\|_{\L2}
        \|F(\phi_2)-F(\phi_1)\|_{\H}.
\end{align*}
Notice that in handling the second integral at the right-hand side we have used
Cauchy-Schwarz inequality in $\L2$. Thus
\begin{multline*}
    \left(1-\beta_2^2w^2\right)\|F(\phi_2)-F(\phi_1)\|_{\H}^2\leq \\
        w^2\gmax L_\Gamma C_P\|\sigma_2-\sigma_1\|_{\L2}
        \|F(\phi_2)-F(\phi_1)\|_{\H},
\end{multline*}
whence, recalling that $\beta_2w\leq\beta w<1$ and dividing by
$\|F(\phi_2)-F(\phi_1)\|_{\H}$ if $\phi_1\ne\phi_2$,
$$ \|F(\phi_2)-F(\phi_1)\|_{\H}\leq
    \frac{w^2\gmax L_\Gamma C_P}{1-\beta_2^2w^2}
    \|\sigma_2-\sigma_1\|_{\L2}. $$
Furthermore, by Morrey inequality and using $\sigma_1,\,\sigma_2\in\U\subset\C$
we get
$$ \|F(\phi_2)-F(\phi_1)\|_\infty\leq
    \frac{w^2\gmax L_\Gamma C_P}{1-\beta_2^2w^2}\|\sigma_2-\sigma_1\|_\infty. $$
But
$$ \vert\phi_2-\phi_1\vert=\vert fF(\phi_2)-fF(\phi_1)\vert\leq
    \Lf\vert F(\phi_2)-F(\phi_1)\vert, $$
so that we finally have
$$ \|A_2(\sigma_2)-A_2(\sigma_1)\|_\infty=\|\phi_2-\phi_1\|_\infty \\
    \leq\frac{w^2\gmax L_\Gamma C_P\Lf}{1-\beta_2^2w^2}
        \|\sigma_2-\sigma_1\|_\infty $$
as desired.

\item Lipschitz continuity of $A$ follows immediately by composition. \qedhere
\end{enumerate}
\end{proof}

\begin{proposition}
The operator $A_1:\V\to\U$ is compact, and so is $A:\V\to\V$.
\end{proposition}
\begin{proof}
\begin{enumerate}
\item Continuity of $A_1$ is implied by Proposition \ref{lipschitz-A12}, hence
in order to have compactness we need to prove that $A_1(\V)$ is relatively
compact in $\U$. We do it by showing that for any $\{\varphi_n\}\subset\V$ the
sequence $\{c_n\}=\{A_1(\varphi_n)\}\subset\U$ contains a convergent
subsequence.

Due to $\|\varphi_n\|_\infty,\,\|c_n\|_\infty\leq 1$, Eq. \eqref{bvp-c} yields
$\|\partial_x^2c_n\|_\infty\leq\alpha w^2$ for all $n$, hence, using the
boundary condition at $x=1$,
$$ \vert\partial_xc_n(x)\vert\leq\intgrl_x^1\vert
    \partial_x^2c_n(\xi)\vert\,d\xi\leq\alpha w^2. $$
In view of these estimates, we conclude that the sequence $\{c_n\}$ is uniformly
bounded and equicontinuous. Owing to Ascoli-Arzel\`a compactness criterion, we
can therefore extract a subsequence $\{c_{n_k}\}$ converging uniformly on
$\bar\domain$ to some $c\in\overline{A_1(\V)}$:
$$ \lim_{k\to\infty}\|c_{n_k}-c\|_\infty=0, $$
whence the compactness of $A_1$ follows.

\item To show the compactness of $A$ we rely on that $A=A_2\circ A_1$.
Continuity of $A$ is implied by Proposition \ref{lipschitz-A12}. Moreover, if
$\{\varphi_n\}$ is a sequence in $\V$ and we let $c_n=A_1(\varphi_n)$,
$\phi_n=A_2(c_n)=A(\varphi_n)$, the compactness of $A_1$ previously proved
allows us to assume, passing to a subsequence if necessary, that $\{c_n\}$
converges in $U$. But then the continuity of $A_2$ implies that also
$\{\phi_n\}$ converges in $\V$, and we have the thesis. \qedhere
\end{enumerate}
\end{proof}

\begin{theorem}
Let Assumptions \ref{assF}, \ref{assg}, \ref{assGamma} hold with $\beta w<1$,
where $\beta=\beta(\ep)$ is the constant defined by Eq. \eqref{beta}.
Then there exists a solution $(\phi,\,c)\in\V\times\U$ to Problem \eqref{bvp}
satisfying the a priori estimate
$$ \|1-c\|_\infty\leq\alpha w^2\|\phi\|_{\L2}. $$
\end{theorem}
\begin{proof}
We know from Proposition \ref{compact-A12} that $A:\V\to\V$ is compact. We claim
now that $\V$ is convex and closed in $\C$.
\begin{enumerate}
\item Convexity: Let $u,\,v\in\V$, $\lambda\in[0,\,1]$, and define
$$ s_\lambda(x)=\lambda u(x)+(1-\lambda)v(x), \quad x\in\bar\domain. $$
Since $u,\,v$ are continuous so is $s_\lambda$, and moreover:
\begin{enumerate}
\item [(i)] $s_\lambda(x)\geq\lambda\ep+(1-\lambda)\ep=\ep$
\item [(ii)] $s_\lambda(x)\leq\lambda+1-\lambda=1$
\end{enumerate}
for all $x\in\bar\domain$, hence $s_\lambda\in\V$ for all $0\leq\lambda\leq 1$.

\item Closure: Let $\{u_n\}\subset\V$ be a sequence converging to some $u\in\C$,
that is $\|u_n-u\|_\infty\to 0$ when $n\to\infty$. Then $\{u_n\}$ converges
pointwise to $u$ on $\bar\domain$, hence we must have
$$ u(x)-\ep=\lim_{n\to\infty}\left(u_n(x)-\ep\right)\geq 0, $$
and analogously
$$ 1-u(x)=\lim_{n\to\infty}\left(1-u_n(x)\right)\geq 0. $$
We conclude therefore $\ep\leq u(x)\leq 1$ for all $x\in\bar\domain$, that
is $u\in\V$.
\end{enumerate}

According to Schauder Fixed Point Theorem, $A$ has a fixed point $\phi\in\V$.
Setting $c=A_1(\phi)$, we deduce that the pair $(\phi,\,c)\in\V\times\U$ solves
Problem \eqref{bvp}.

The \emph{a priori} estimate on $\phi,\,c$ readily follows from Lax-Milgram
Theorem applied to Problem \ref{wellpos-c} with $\varphi=\phi$.
\end{proof}

\begin{theorem}
Let $\phi\in\V$ be any solution to Problem \eqref{bvp}. Then
$$ \|\phi-\phi_0\|_\infty\leq\frac{\gmax}{\Lg C_P}(\beta_2w)^2. $$
\end{theorem}
\begin{proof}
>From Proposition \ref{wellpos-phi} we know
$$ \|F(\phi)-F(\phi_0)\|_{\H}\leq w^2C_P\gmax\Gmax. $$
Rearranging the coefficients according to the definition of $\beta_2$ given by
Eq. \eqref{beta12} we get
$$ \Lg\Lf\|F(\phi)-F(\phi_0)\|_{\H}\leq\frac{\gmax}{C_P}(\beta_2w)^2, $$
whence the thesis follows using Morrey inequality at the left-hand side and then
considering that
$$ \vert\phi(x)-\phi_0\vert\leq\Lf\vert F(\phi(x))-F(\phi_0)\vert $$
for all $x\in\bar\domain$ because $\phi\geq\ep$ on $\bar\domain$.
\end{proof}

\begin{theorem}
Let Assumption \ref{assGamma-2} hold and assume that a solution
$(\phi,\,c)\in\V\times\U$ exists to Problem \eqref{bvp}, such that $\phi$
fulfills the free boundary condition \eqref{bc-freebound}. Then
$\phi(x)\geq\phi_0$ for all $x\in\bar\domain$.
\end{theorem}
\begin{proof}
\begin{enumerate}
\item Let us integrate once the differential equation for $\phi$ on
$\bar\domain$. Owing to the boundary condition $\partial_x\phi(0)=0$ and to the
free boundary condition $\partial_x\phi(1)=0$ we find
$$ \intgrl_0^1 g(\phi)\Gamma(c)\,dx=0. $$
Since $g(\phi)\geq 0$ is continuous and cannot identically vanish on
$\bar\domain$ (basically because the constant $\phi\equiv 1$ does not solve
Problem \eqref{bvp}, as it does not match the boundary condition
$\phi(1)=\phi_0<1$), we deduce that either $\Gamma(c)\equiv 0$ on $\bar\domain$
or $\Gamma(c)$ changes sign on $\domain$. Observing that, due to the properties
of $c$ (cf. Proposition \ref{wellpos-c}) and of $\Gamma$ (cf. Assumption
\ref{assGamma-2}), the function $x\mapsto \Gamma(c(x))$ is continuous and
nonincreasing on $\bar\domain$ with $\Gamma(c(0))=\Gamma(1)>0$, we conclude that
there must exist $\bar x\in\domain$ such that $\Gamma(c)>0$ in $[0,\,\bar x)$
and $\Gamma(c)\leq 0$ in $[\bar x,\,1]$.

\item From the equation
$$ -\partial^2_xF(\phi)=w^2g(\phi)\Gamma(c) \quad \text{in\ } \domain $$
we discover that
$$ \partial_x^2F(\phi)\ \text{is\ }
    \begin{cases}
        \leq 0 & \text{for\ } x\in[0,\,\bar x) \\
        \\[-0.3cm]
        \geq 0 & \text{for\ } x\in[\bar x,\,1],
    \end{cases}
$$
hence that $F(\phi)$ is concave in the interval $[0,\,\bar x]$ and convex in
the interval $(\bar x,\,1]$. Since $\partial_xF(\phi)=F'(\phi)\partial_x\phi$,
we further obtain $(\partial_xF(\phi))(0)=0$, whence $\partial_xF(\phi)\leq 0$
in $[0,\,\bar x)$, and $(\partial_xF(\phi))(1)=0$, thus $\partial_xF(\phi)\leq
0$ also in $[\bar x,\,1]$. In conclusion, we have $\partial_xF(\phi)\leq 0$ on
the whole $\bar\domain$.

\item Assume now by contradiction that $(F(\phi))(x)<F(\phi_0)$ for some
$x\in[0,\,1)$. Then
$$ (F(\phi))(1)=(F(\phi))(x)+\intgrl_x^1\partial_xF(\phi)\,d\xi<F(\phi_0), $$
which however is incompatible with the boundary condition $\phi(1)=\phi_0$ of
Problem \eqref{bvp}. Hence $F(\phi)\geq F(\phi_0)$ on $\bar\domain$.

\item Finally, we see that it is impossible to have $\phi(x)<\phi_0$ at any
$x\in\bar\domain$, for this would imply $(F(\phi))(x)<F(\phi_0)$ which
contradicts the previous result. Therefore $\phi\geq\phi_0$ on $\bar\domain$
and we have the thesis. \qedhere
\end{enumerate}
\end{proof}

\begin{proposition}
Let $\czero_w\in\U$ be the solution to Problem \eqref{bvp-c0}. Then:
\begin{enumerate}
\item [(i)] The mapping $w\mapsto\czero_w(x)$ for $w\in\R_+$, $x\in\bar\domain$
is continuous and nonincreasing.

\item [(ii)] For $w\to+\infty$ the functions $\czero_w$ tend pointwise to the
function
$$ \czero_\infty(x)=
    \begin{cases}
        1 & \text{if\ } x=0 \\
        \\[-0.3cm]
        0 & \text{if\ } x\ne 0.
    \end{cases}
$$

\item [(iii)] For all $c_0\in(0,\,1)$ there exists $w_\ast>0$ such that if
$w\geq w_\ast$ the equation
$$ \czero_w(x)=c_0 $$
in the unknown $x$ has exactly one solution $\bar x_w\in\bar\domain$.

\item [(iv)] For any fixed $c_0\in(0,\,1)$ one has $\bar x_w\to 0$ when
$w\to+\infty$.
\end{enumerate}
\end{proposition}
\begin{proof}
Some of these properties are more easily proved by referring to Problem
\eqref{bvp-c0} rather than to its explicit solution \eqref{czero}.
\begin{enumerate}
\item To show continuity of the mapping $w\mapsto\czero_w(x)$, let us fix
$w_1,\,w_2\geq 0$ and set $u=\czero_{w_2}-\czero_{w_1}\in\C$. Then $u$ solves
the problem
\begin{equation}
    \left\{
    \begin{array}{rcll}
        -\partial_x^2u+\alpha w_2^2\phi_0u & = &
            -\alpha(w_2^2-w_1^2)\phi_0\czero_{w_1} & \text{in\ } \domain \\
        \\[-0.3cm]
        u & = & 0 & x=0 \\
        \\[-0.3cm]
        \partial_xu & = & 0 & x=1
    \end{array}
    \right.
    \label{continuity-w}
\end{equation}
so that, owing to Lemma \ref{lemma-uhk}, we have
$$ \|\czero_{w_2}-\czero_{w_1}\|_\infty=\|u\|_\infty\leq
    \alpha\phi_0\vert w_2^2-w_1^2\vert\|\czero_{w_1}\|_\infty\leq
    \alpha\phi_0\vert w_2^2-w_1^2\vert. $$
Consequently $\czero_{w_2}(x)\to\czero_{w_1}(x)$ for all $x\in\bar\domain$ as
$w_2\to w_1$, and continuity follows.

\item We address now monotonicity of $w\mapsto\czero_w(x)$. Assuming $w_1\leq
w_2$, multiply the differential equation in \eqref{continuity-w} by
$u^+\in H^1_{0,0}(\domain)$, then integrate by parts over $\domain$:
$$ \|u^+\|_{\H}^2+\alpha w_2^2\phi_0\|u^+\|_{\L2}^2=
    -\alpha(w_2^2-w_1^2)\phi_0\intgrl_0^1\czero_{w_1}u^+\,dx\leq 0. $$
We deduce $\|u^+\|_{\H}=0$, whence $u\leq 0$ on $\bar\domain$ and
finally $\czero_{w_2}(x)\leq\czero_{w_1}(x)$ for all $x\in\bar\domain$, thus
completing the proof of (i).

\item Regarding (ii), it is customary to use the expression of $\czero_w$ given
by Eq. \eqref{czero}. In particular, it is evident that we must have
$\czero_\infty(0)=1$ because $\czero_w(0)=1$ for all $w\geq 0$ due to the
boundary conditions of Problem \eqref{bvp-c0}. Moreover, if $x\in(0,\,1]$ then
$$ \czero_w(x)\sim e^{-w\sqrt{\alpha\phi_0}x} \quad (w\to +\infty), $$
therefore $\czero_w(x)\to 0$ for $w\to +\infty$.

\item Fix now any $c_0\in(0,\,1)$. We observe that for $w=0$ we have
$\czero_0(x)\equiv 1$ while, as a consequence of (ii), for $w\to+\infty$ it
results $\czero_w(1)\to 0$. Since we know from (i) that $w\mapsto\czero_w(1)$ is
continuous, we conclude that $\czero_w(1)$ takes all values in $(0,\,1]$ as $w$
ranges in $\R_+$, hence there exists $w_\ast>0$ such that
$\czero_{w_\ast}(1)=c_0$. But (i) also tells us that $w\mapsto\czero_w(1)$ is
nonincreasing, thus $\czero_w(1)\leq c_0$ for all $w\geq w_\ast$. Finally, since
$x\mapsto\czero_w(x)$ is continuous on $\bar\domain$ with $\czero_w(0)=1>c_0$,
for all $w\geq w_\ast$ there must exist a point $\bar x_w\in\bar\domain$ such
that $\czero_w(\bar x_w)=c_0$. Uniqueness of $\bar x_w$ follows from the strict
monotonicity of $\czero_w(x)$ with respect to $x\in\bar\domain$ (use Proposition
\ref{wellpos-c} along with $\czero_w(x)>0$ for all $x\in\bar\domain$ to
discover $\partial_x\czero_w(x)<0$ over $\bar\domain$). This gives (iii).

We observe that, under the hypothesis $w\geq w_\ast$, it results
$\czero_w(x)\geq c_0$ for $x\in[0,\,\bar x_w]$ and $\czero_w(x)<c_0$ for
$x\in(\bar x_w,\,1]$.

\item Finally, let us consider statement (iv). First we check that the limit of
$\bar x_w$ for $w\to+\infty$ exists by observing that for $w_2\geq
w_1\geq w_\ast$ it results
$$ \czero_{w_1}(\bar x_{w_2})\geq \czero_{w_2}(\bar x_{w_2})=c_0=
    \czero_{w_1}(\bar x_{w_1}), $$
where we have used the monotonicity of $w\mapsto\czero_w$ and the definition
of $\bar x_{w_1},\,\bar x_{w_2}$ respectively. Since $x\mapsto\czero_{w_1}(x)$
is nonincreasing, this says $\bar x_{w_2}\leq \bar x_{w_1}$, hence
$w\mapsto\bar x_w$ is monotone nonincreasing and admits therefore a limit at
$+\infty$. Also notice that
\begin{equation}
    \lim_{w\to+\infty}\bar x_w=\inf_{w\geq w_\ast}\bar x_w\in\bar\domain.
    \label{limit-xw}
\end{equation}
Assume now by contradiction that $\lim_{w\to+\infty}\bar x_w=\xi>0$. Fix then
$w\geq w_\ast$ and consider the point $x=\xi\in\domain$. Owing to
\eqref{limit-xw} we have $\xi\leq\bar x_w$, whence $\czero_w(\xi)\geq
\czero(\bar x_w)=c_0$, for all $w\geq w_\ast$, so
$$ \czero_\infty(\xi)=\lim_{w\to+\infty}\czero_w(\xi)\geq c_0>0, $$
which is in contrast with (ii). Thus $\bar x_w\to 0$ and we obtain (iv).
\qedhere
\end{enumerate}
\end{proof}

\begin{theorem}
Let Assumptions \ref{assGamma}, \ref{assGamma-2} hold. Then there exists a
positive solution to Eq. \eqref{rootsGamma}.
\end{theorem}
\begin{proof}
Let us denote
$$ C(w)=\intgrl_0^1\Gamma(\czero_w)\,dx; $$
since $\czero_0(x)\equiv 1$, we deduce $C(0)=\Gamma(1)>0$ in view of Assumption
\ref{assGamma-2}.
\begin{enumerate}
\item We claim that $C$ is continuous on $\R_+$. Indeed, let $w_1,\,w_2\geq
0$. Using Assumption \ref{assGamma} and Proposition \ref{c0-w} we get
\begin{align*}
    \vert C(w_2)-C(w_1)\vert &\leq \intgrl_0^1\vert\Gamma(\czero_{w_2})-
        \Gamma(\czero_{w_1})\vert\,dx \\
    & \leq L_\Gamma\|\czero_{w_2}-\czero_{w_1}\|_\infty\leq
        L_\Gamma\alpha\phi_0\vert w_2^2-w_1^2\vert,
\end{align*}
so that for $w_2\to w_1$ it results $C(w_2)\to C(w_1)$.
\setcounter{savedenumi}{\value{enumi}}
\end{enumerate}
Due to the positivity of $C(0)$, existence of a positive root of Eq.
\eqref{rootsGamma} is simply obtained by showing that $C(w)<0$ for $w$ large
enough: Continuity will do the rest.

For $c_0\in(0,\,1)$ fixed by Assumption \ref{assGamma-2}, let $w_\ast>0$
be the value of $w$ mentioned by Proposition \ref{c0-w}(iii): We prove that
there exists $w_{\ast\ast}>w_\ast$ such that $C(w)<0$ for all $w\geq
w_{\ast\ast}$. To this end, we will henceforth assume $w\geq w_\ast$ and rewrite
$C$ as
$$ C(w)=\intgrl_0^{\bar x_w}\Gamma(\czero_w)\,dx+
    \intgrl_{\bar x_w}^1\Gamma(\czero_w)\,dx=:C_1(w)+C_2(w). $$
\begin{enumerate}
\setcounter{enumi}{\value{savedenumi}}
\item Notice that $C_1(w)>0$ for all $w\geq w_\ast$, indeed
$\Gamma(\czero_w(x))>0$ for $x\in[0,\,\bar x_w)$. The Mean Value Theorem implies
the existence of $x_w\in(0,\,\bar x_w)$ such that
$$ C_1(w)=\bar x_w\Gamma(\czero_w(x_w)). $$
By consequence, using Proposition \ref{c0-w}(iv) we get $C_1(w)\leq\Gmax
\bar x_w\to 0$ for $w\to+\infty$.

\item Concerning $C_2$, we begin by observing that there exists $\bar w_\ast\geq
w_\ast$ such that $\Gamma(\czero_w(x))\leq 0$ for all $x\in[\bar x_w,\,1]$ and
all $w\geq\bar w_\ast$, with $\Gamma(\czero_w(x))$ not identically zero on
$[\bar x_w,\,1]$. Indeed, we know from Proposition \ref{c0-w}(ii) that the
$\czero_w$'s tend pointwise to zero almost everywhere in $[0,\,1]$ as
$w\to+\infty$, and from Assumption \ref{assGamma-2} that $\Gamma(0)<0$.
Consequently we also deduce $C_2(w)<0$ for all $w\geq \bar w_\ast$.

\item Next we claim that $C_2$ is a nonincreasing function of $w$. Given
$w_2\geq w_1\geq w_\ast$, we use the
monotonicity of $w\mapsto\czero_w$ and of $\Gamma$ to discover indeed
$$ C_2(w_1)=\intgrl_{\bar x_{w_1}}^1\Gamma(\czero_{w_1})\,dx\geq
    \intgrl_{\bar x_{w_1}}^1\Gamma(\czero_{w_2})\,dx\geq
    \intgrl_{\bar x_{w_2}}^1\Gamma(\czero_{w_2})\,dx=C_2(w_2). $$

\item Fix $\ep>0$. Since $C_1$ tends to zero for $w\to+\infty$, there
exists $w_{\ast\ast}\geq \bar w_\ast$ such that $C_1(w)<\ep$ for all $w\geq
w_{\ast\ast}$. Moreover, $C_2(w_{\ast\ast})<0$. Let us now choose
$\ep=\frac{1}{2}\vert C_2(w_{\ast\ast})\vert=
-\frac{1}{2}C_2(w_{\ast\ast})$. If $w\geq w_{\ast\ast}$, we find
$$ C(w)=C_1(w)+C_2(w)<\ep+C_2(w_{\ast\ast})=
    \frac{1}{2}C_2(w_{\ast\ast})<0 $$
and the thesis follows. \qedhere
\end{enumerate}
\end{proof}

\section{Conclusions and perspectives}
\label{sect-conclusions}
In this paper we have developed a macroscopic model which describes the
evolution of a tumor cord using the theory of mixtures. The tumor mass has been
modeled as a three-phasic porous medium composed by cells, extracellular fluid,
and extracellular matrix, growing along a blood vessel whence cells get the
necessary nutrient for their vital functions. As a matter of fact, since the
main interest was in investigating the growth process in connection with the
availability of nutrient, the extracellular matrix has been regarded as a rigid
non-remodeling scaffold in which cells live and proliferate and extracellular
fluid flows, therefore it has not played a primary role in the dynamics of the
system. By means of suitable principles of mixture theory, in particular the
reference to a Darcy-like framework to model interactions among different
components, the problem has been reduced to a system of two partial differential
equations. The first, for the cell volume ratio, is derived from the cell mass
balance in which the velocity is explicitly expressed in terms of the internal
stress state of the cells. It also includes a source/sink term describing
proliferation or death of the cells according to the current size of the cell
population itself, as well as to the available amount of oxygen present in the
mixture. The second, accounting for the dynamics of the nutrient, is a linear
parabolic equation modeling the diffusion of oxygen from the vessel wall toward
the mixture and its contemporaneous uptake by the cells. This simple setting has
to be understood in the light of the experimental evidence that the contribution
of a possible transport by the extracellular fluid to the motility of oxygen
molecules is definitely negligible with respect to their high diffusivity, hence
the corresponding advection terms can be rightly dropped from the equations.
Interaction of the growing tumor cord with a surrounding host tissue composed by
normal non-proliferating cells has also been considered.

The mathematical formalization of these ideas results in a free boundary
problem constituted by a system of two partial differential equations
supplemented by suitable boundary and interface conditions. The free boundary
is represented in this context by the interface separating the cord from the
host tissue, which changes in time according to the growth of the tumor in
interaction with the external normal cells and has to be determined as a
further unknown of the problem.

The model is intrinsically multidimensional. In addition, it is not linked to
a particular geometry because it does not assume any special configuration of
the domain. In view of this, it is in principle applicable to a wide range of
different scenarios. In this work, a two-dimensional problem has been
specifically considered, focusing on the axial growth of the cord along the
blood vessel and on its simultaneous expansion outward the vessel. Numerical
simulations have shown ability of the model to reproduce specific features of
vascular tumors, like the formation of an inner vital zone, constituted by
sufficiently fed cells that duplicate, and an outer necrotic zone at the
periphery of the cord, where conversely cells are not reached by a proper
quantity of nutrient and cannot hence proliferate nor survive. It is worth
pointing out that some different models of tumor growth assume \emph{a priori}
the existence of such zones and describe them invoking specific equations. In
the present model, their formation is instead recovered \emph{a posteriori} as a
consequence of the overall dynamics predicted by a unique set of equations, by
appealing to relatively general guidelines. Another interesting feature made
evident by the numerical results concerns the evolution of the system toward a
steady state. The cord clearly exhibits two different behaviors in the front
part, the head, which is vital since always globally fed by a sufficient amount
of oxygen, and in the rear part, the tail, which elongates as the head moves
forward along the blood vessel, keeping a rectilinear shape with a nearly
constant width in the transverse direction.

These considerations suggest that, as far as axial growth is concerned, two
main questions have to be addressed, namely the dynamics of the head in terms
of shape and velocity of propagation and the equilibrium of the tail in terms
of maximum penetration inside the host tissue. In the present work, we have
mathematically investigated this second issue, studying existence and
regularity of physically significant solutions that describe the distribution
of cells and nutrient across the tail at the steady state. By addressing
furthermore the free boundary problem via a perturbative expansion of the
solution, we have also been able to characterize qualitatively and
quantitatively the steady growth width of the cord.

The model developed in this paper should be regarded as minimal, in the sense
that it is deduced from few outline principles of a general well-coded theory,
and takes into account only those really fundamental aspects to obtain a
qualitative mathematical and biological description of the macroscopic
phenomenon under consideration. Of course, it is liable to many improvements,
which still at a basic level may involve, for instance, a more accurate
description of both host cells and extracellular matrix, that here are
essentially passive and lacking in their own physiology, and a more
sophisticated coupling between the dynamics of the tumor cells and the nutrient,
accounting for different uptake rates on the basis of the specific functions
(survival, proliferation) carried out by cells. However, we believe that a
minimal model, for which both a theoretical comprehension and validation are
possible in view of its relatively simple structure, constitutes a solid
foundation, hence a safe starting point, to tackle more complicated situations.

\section*{Acknowledgements}
The author is deeply indebted to prof. Luigi Preziosi, who had the early idea
of this research. This work would have never been possible without his
invaluable suggestions and his constant support. Furthermore, the author wants
to address special thanks to prof. Paolo Tilli for his kind friendship
demonstrated during several useful discussions about the contents of Sects.
\ref{sect-statprob}, \ref{sect-proofs} and in reading their first proofs.

\providecommand{\bysame}{\leavevmode\hbox
to3em{\hrulefill}\thinspace}
\providecommand{\MR}{\relax\ifhmode\unskip\space\fi MR }
\providecommand{\MRhref}[2]{%
  \href{http://www.ams.org/mathscinet-getitem?mr=#1}{#2}
} \providecommand{\href}[2]{#2}

\end{document}